\documentclass[article,ij4uq]{preprint}              %  final version
\usepackage{fixfoot}
\usepackage{booktabs}
\usepackage{makecell}
\usepackage{import}
\usepackage{listings}
\usepackage{tkz-graph}
\usepackage{algpseudocode}
\usepackage{float}
\usepackage{epstopdf}
\usepackage{mathtools}
\usepackage{hyperref}
\usepackage{upref}
\usepackage{amscd,amsxtra}
\usepackage{dsfont}

\usepackage[hang]{footmisc}
\fancypagestyle{plain}{%
  \fancyhf{}
  \fancyfoot[R]{\small\bf\thepage }
  }
\renewcommand{\dmyy}{19}

\renewcommand{\vec}{\bm}

\newcommand{\tf}{t_{f}}
\def\<#1>{\langle{#1}\rangle}

\long\def\comment#1{}

\def\bA{\bm{A}}

\def\bx{\bm{x}}

\DeclarePairedDelimiter\floor{\lfloor}{\rfloor}


\newtheorem{proposition}{Proposition}[section]

\begin{document}
%\title{Multifidelity ST-MCMC for Bayesian parameter inference and model comparison from discrete single-cell measurements}
\title{Bayesian inference of Stochastic reaction networks using Multifidelity Sequential Tempered Markov Chain Monte Carlo}
%For at least  authors with different addresses, use instead the following commands
\corrauthor[1]{Thomas A. Catanach}
\author[2]{Huy D. Vo}
\author[2]{Brian Munsky}
\corremail{tacatan@sandia.gov}
\address[1]{Sandia National Laboratories, Livermore, CA}
\address[2]{Dept. of Chemical and Biological Engineering, Colorado State University, Fort Collins, CO}

\abstract{Stochastic reaction network models are often used to explain and predict the dynamics of gene regulation in single cells. These models usually involve several parameters, such as the kinetic rates of chemical reactions, that are not directly measurable and must be inferred from experimental data.
 Bayesian inference provides a rigorous probabilistic framework for identifying these parameters by finding a posterior parameter distribution that captures their uncertainty.
 Traditional computational methods for solving inference problems such as Markov Chain Monte Carlo methods based on classical Metropolis-Hastings algorithm involve numerous serial evaluations of the likelihood function, which in turn requires expensive forward solutions of the chemical master equation (CME).
 We propose an alternative approach based on a multifidelity extension of the Sequential Tempered Markov Chain Monte Carlo (ST-MCMC) sampler.
 This algorithm is built upon Sequential Monte Carlo and solves the Bayesian inference problem by decomposing it into a sequence of efficiently solved subproblems that gradually increase model fidelity and the influence of the observed data.
 We reformulate the finite state projection (FSP) algorithm, a well-known method for solving the CME, to produce a hierarchy of surrogate master equations to be used in this multifidelity scheme.
 To determine the appropriate fidelity, we introduce a novel information-theoretic criteria that seeks to extract the most information about the ultimate Bayesian posterior from each model in the hierarchy without inducing significant bias.
 This novel sampling scheme is tested with high performance computing resources using biologically relevant problems.
}

\maketitle

\section{Introduction}
A distinguishing feature of biology is the diversity manifested by living things across different scales, from the readily observed multitude of species to the differences between individuals of the same species.
At the microscopic level, a population of cells with the same genetic code, growing under the same lab conditions, could still display phenotypic variability in gene products~\cite{Elowitz2000, Munsky2009, Neuert2013a, Singh2014, Bruggeman2017}.
Phenotypic variability has been observed in an increasing volume of data obtained from single-cell, single-molecule measurements enabled by recent progresses in chemical labeling and imaging techniques~\cite{Raj2008, Munsky2015, Li2019}.

Much of the variability in gene expression is attributed to the stochasticity of vital cellular processes (e.g., transcription, translation) that are subjected to the randomness of molecular interactions.
Stochastic reaction networks (SRN) represent a class of models that have been widely used to capture  temporal and spatial fluctuations in single-cell gene expression~\cite{Gillespie1992}.
SRN models treat the copy numbers of biochemical species, i.e. the number of molecules of a given type within a cell, as states in a discrete-space, continuous-time Markov process, where chemical reactions are represented by transitions between states.
Given an SRN model, the probabilities of gene expression states within a cell can be computed by solving the chemical master equation (CME). This is a dynamical system in an infinite-dimensional space that describes the evolution of the probability distribution of all states.
The finite state projection (FSP) is a well-known approximation method to obtain high-fidelity solutions of the CME~\cite{Munsky2006}.
This method reduces the intractable state space of the original SRN into a finite subset chosen based on a proven error bound, turning the infinite-dimensional CME system into a finite problem of linear differential equations.

The present work is concerned with the selection, parameter estimation, and uncertainty propagation of these reaction network models within the Bayesian framework.
 Bayesian methods are a powerful tool for system identification for SRN models because they provide rigorous uncertainty quantification by identifying a probability distribution over plausible model parameters instead of selecting a single model that may fit the data well~\cite{Gomez-Schiavon2017, Schnoerr2017, Tiberi2018, Weber2018, Ting2019}.
 This distribution over the models given the data is called the posterior distribution.
 Quantifying model parameter uncertainty is critical because it is difficult to model the full complexity of the biological system and the biological system may exhibit experimental context dependence~\cite{Catanach2018Context}.
 Further, once model parameter uncertainty has been quantified, further experiments can be designed to provide new information about the system~\cite{Lindley2007, Huan2013, Ruess2013, Fox2019, Busetto2013}.

In this paper, we focus on data obtained by the experimental technique of single-molecule fluorescent in situ hybridization (smFISH). These datasets consist of independent single-cell measurements, each of which measures the copy number of biochemical species at a single time point.
The standard approach to sample from the posterior distribution implied by this data is to use Markov Chain Monte Carlo (MCMC) algorithms such as the random walk Metropolis-Hastings MCMC sampler~\cite{Gomez-Schiavon2017, Vo2019a}.
With high-fidelity CME solutions enabled by the FSP, one can compute the likelihood of observing these single-cell data and then perform Bayesian inference for model parameters.
However, this approach suffers from two drawbacks. The first drawback is that the MCMC is inherently serial, preventing it from utilizing the massively parallel processing capability provided by modern high performance computing clusters. Therefore, if standard MCMC techniques are used, several tens to hundreds of thousands of sequential model evaluations may be needed to adequately sample the posterior distribution. The second drawback is that the FSP solutions required for the likelihood function are expensive, often requiring several minutes per model evaluation for a moderately sized problem. Typically, the number of differential equations that the FSP algorithm needs to solve grows exponentially with the number of species in the network, so the size of the state space and transition matrix quickly grows intractable.
This has motivated Approximate Bayesian Computation (ABC) approaches that replace the computationally expensive single-cell likelihood function with less expensive model-data discrepancy functions~\cite{Wu2014, Prescott2018, Warne2019}.
Samples produced by ABC are, in general, not distributed according to the true posterior distribution, and a careful choice of summary statistics is critical for the performance and reliability of ABC samplers.

We propose in this work a different approach that uses a parallel and multifidelity MCMC computational framework to produce samples from the true posterior distribution. Many approaches to parallel MCMC methods have been proposed either based on parallel proposals or parallel Markov chains \cite{rosenthal2000parallel, Liu2000, ChingChen2007, Wu2017}. One family of popular parallel MCMC methods are those based upon sequential Monte Carlo (SMC) samplers~\cite{Moral2006, ChingChen2007, Kantas2014, catanach2018bayesian}. For this work, we replace the standard MCMC methods with the Sequentially Tempered Markov Chain Monte Carlo (ST-MCMC)~\cite{catanach2018bayesian}, which is a massively parallel sampling scheme based on SMC. This method transports a population of model parameter samples through a series of intermediate annealing levels that reflect the gradual increase in the influence of the data likelihood. At each level, MCMC is used to explore the intermediate distribution and re-balance the distribution of the samples. Since this method is population based, these MCMC steps can be done in parallel. Further, this method can effectively adapt to the target posterior distribution to speed up sampling.

Similarly, approaches to multifidelity MCMC have been explored in the literature such as multifidelity delayed acceptance schemes, Multilevel Markov Chain Monte Carlo, and multifidelity approaches to SMC~\cite{AndresChristen2005, Efendiev2006, Cui2015, Koutsourelakis2009, Dodwell2015, Latz2018}. Multifidelity delayed acceptance schemes have been applied to Bayesian inference for the CME before~\cite{Golightly2015, Sherlock2016a, Vo2019a}. Within these methods, a fast surrogate of the expensive likelihood function is used to pre-screen proposed samples within MCMC before they are accepted or rejected based on the expensive CME likelihood. This method still requires many sequential full model evaluations in order to sample the posterior. Multilevel MCMC~\cite{Dodwell2015} uses a hierarchy of models, such as different discretization grids of a PDE, to design an estimator for a specific quantify of interest. This uses parallel Markov chains at different model fidelities to estimate a correction to the quantify of interest estimate incurred by refining the model fidelity. Multifdelity SMC methods like the Multilevel Sequential$^2$ Monte Carlo sampler~\cite{Latz2018} use an embarrassingly parallel approach and a hierarchy of multifidelity models. Therefore, only a few sequential full model evaluations may be needed. We take this approach to develop a multifidelity form of ST-MCMC for solving the CME.  Within our method, instead of only considering a series of annealing levels, we also consider a hierarchy of model fidelity.  Thus, the solution of the full inference problem is broken down by steering the samples from a distribution that reflects a low fidelity model with little influence from the likelihood to a distribution that reflects the high-fidelity model with the full influence of the likelihood. By performing the early updates using fast models, the sampler can quickly converge to the most important regions of parameter space, where a higher fidelity model can then be used to better assess which regions of this high-likelihood space are most likely. The key challenge for applying multifidelity methods in the SMC context is deciding which annealing factor and model fidelity is appropriate at a given level. Latz et al. suggest an approach based on the effective sample size of the population~\cite{Latz2018}. We take a different approach by leveraging a limited number of high fidelity model solves to estimate the information gained about the ultimate posterior given a current model fidelity and annealing factor. This information theoretic criteria can effectively identify when the lower fidelity model is overly biasing the solution and should be discarded in favor of a higher fidelity model.

We demonstrate the efficiency and accuracy of our novel scheme when solving parameter estimation, model selection, and uncertainty propagation for stochastic chemical kinetic models in a Bayesian framework. This new approach to multifidelity ST-MCMC using fast surrogates from reduced order models based on a novel reformulation of the FSP significantly reduces the number of expensive likelihood function required to sample the posterior. The example problems are based on models from the system and synthetic biology literature. These include a three-dimensional repressilator gene circuit, a spatial bursting gene expression network, and a stochastic transcription network for the inflammation response gene IL1beta~\cite{Kalb2019}. As such the primary contributions of this paper are:
\begin{enumerate}
\item Development of multifidelity ST-MCMC for Bayesian inference of SRNs
\item Introduction of an information-theoretic criteria for assessing the appropriate model fidelity within multifidleity ST-MCMC
\item Description of a novel surrogate model of the CME to be used within multifidelity inference problems
\end{enumerate}

This paper is organized as follows: Section~\ref{sec:background} describes general background for stochastic reaction network modeling, finite state projection, Bayesian inference, and MCMC methods. Second~\ref{sec:multifi_stmcmc} describes the  Multifidelity Sequential Tempered MCMC algorithm and the information theoretic criteria for adapting model fidelity. Section~\ref{sec:cme_surrogate} describes a novel surrogate models for the CME. Section~\ref{sec:numerical} describes three experiments to identify model parameters in SRNs using Multifidelity ST-MCMC. Finally, Section~\ref{sec:conclusion} concludes.

\section{Background}
\label{sec:background}
\subsection{Stochastic reaction networks for modeling gene expression}
A reaction network consists of $N$ different chemical species $S_1, \ldots, S_N$
that are interacting via the following $M$ chemical reactions
\begin{equation}
  \label{eq:reaction}
  \nu_{1,j}^{react}S_1 + \ldots + \nu_{N,j}^{react}S_N
  \longrightarrow
  \nu_{1,j}^{prod}S_1 + \ldots + \nu_{N,j}^{prod}S_N
  .
\end{equation}
We are interested in keeping track of the integral vectors
$\bm{x} \equiv (x_1, \dots, x_N )^T$, where $x_i$ is the
population
of the $i$th species.
Assuming constant temperature and volume, the time-evolution of this system can be modeled by a continuous-time, discrete-space Markov process~\cite{Gillespie1992}.
The $j$th reaction channel is associated with a stoichiometric vector
$\bm{\nu}_j = (\nu_{1,j}^{prod} - \nu_{1,j}^{react}, \ldots, \nu_{N,j}^{prod} - \nu_{N,j}^{react})^T$ ($j=1,\ldots,M$) such that, if the
system is in state $\bm{x}$ and reaction $j$ occurs, the system
transitions to state $\bm{x}+\bm{\nu}_j$.
Given $\bm{x}(t)=\bm{x}$, the propensity $\alpha_j(\bm{x}; \theta)dt$ determines the probability that reaction $j$ occurs in the next infinitesimal time
interval $[t,t+dt)$, where $\theta$ is the vector of model parameters. In other words,
$$\operatorname{Prob}(\bm{x}(t + dt) = \bm{x} + \bm{\nu}_j | \bm{x}(t) = \bm{x}) = \alpha_j(\bm{x};\theta)dt.$$

An important case of reaction networks are those that follow mass-action kinetics, whose propensity functions take the form
\begin{equation}
\label{eq:propensity_mass_action}
\alpha_j(\bm{x}; \theta)
=
c_j(\theta)
\binom{x_1}{\nu_{1,j}^{react}} \cdot \ldots \cdot \binom{x_N}{\nu_{N,j}^{react}}.
\end{equation}
In this formulation, $c_j(\theta)$ is the probability per unit time that a combination of molecules can react via reaction $j$, and the remaining factor is the number of ways the existing molecules can be combined to form the left side of the chemical equation~\eqref{eq:reaction}.

 The time-evolution of the probability distribution of this  Markov process is the solution of the linear system of differential equations known as the chemical master equation (CME)
\begin{equation}
\label{eq:cme_ode_form}
\begin{cases}
\frac{d}{dt}\bm{p}(t)=\bm{A(\theta)}\bm{p}(t),\quad t\in[0,\,t_f]\\
\bm{p}(0)=\bm{p}_0
\end{cases},
\end{equation}
where $\bm{p}(t)$ is the time-dependent probability distribution of all states, $p(t, \bx)=\textrm{Prob}\{\bm{x}(t)=\bm{x}_i \vert \bm{x}(0)\}$.
The initial distribution $\bm{p}_0$ is assumed to be given, and $\bm{A}(\theta)$ is the infinitesimal generator of the Markov process, defined entry-wise as
\begin{equation}
  \label{eq:cme_matrix}
  \bm{A}(\bm{y}, \bm{x}; \theta)
  =
  \begin{cases}
    \alpha_j(\bm{x}; \theta) \; \text{ if } \bm{y} = \bm{x} + \bm{\nu}_j \\
    -\sum_{j=1}^{M}{\alpha_j(\bm{x}; \theta)} \; \text{ if } \bm{y} = \bm{x} \\
    0 \text{ otherwise}
  \end{cases}.
\end{equation}
Here, we have made explicit the dependence of $\bA$ on the model parameter vector $\theta$, which we need to infer from experimental data.
\subsection{The finite state projection}
\label{sec:fsp_background}
Typically, reaction networks model open biochemical systems, where the set of all possible molecular states is unbounded. This makes the CME an infinite-dimensional linear system of ODEs.
The finite state projection (FSP) is a well-known strategy to systematically reduce this linear system into a finite surrogate model with a strict error bound.

The FSP can be thought of as a special class of projection-based model reduction applied to the CME. Specifically, let $\Omega$ be a finite subset of the CME state space. The projection of the CME operator $\bm{A}$ onto the subspace spanned by the point-mass measures $\{ \delta_{\bm{x}} \vert \bm{x} \in \Omega\}$ is given by
\begin{equation}
  \widetilde{\bm{A}}_{\Omega}(\bm{y}, \bm{x}) =
  \begin{cases}
    \bm{A}(\bm{y}, \bm{x}) \; \text{if } \bm{x}, \bm{y} \in \Omega \\
    0 \; \text{otherwise}
  \end{cases}
  .
\end{equation}

We can then define a reduced model of the dynamical system~\eqref{eq:cme_ode_form} based on this projection as
\begin{equation}
  \label{eq:fsp_system}
  \frac{d}{dt}{\widetilde{\bm{p}}_{\Omega}}(t)
  =
  \widetilde{\bm{A}}_{\Omega}\widetilde{\bm{p}}_{\Omega}(t),
  \;
  t \in [0, \tf]
  .
\end{equation}
Clearly, in solving~\eqref{eq:fsp_system} we only need to keep track of the equations corresponding to states in the finite set $\Omega$, which is amenable to numerical treatments.

In contrast to generic projection methods, the gap between the reduced-order model and the true CME can be computed for the FSP. Indeed, Munsky and Khammash ~\cite[Theorem 2]{Munsky2008} proved that the truncation error can be quantified in $\ell_1$ norm as
\begin{equation}
  \label{eq:fsp_error_bound}
  \|\bm{p}(t) - \widetilde{\bm{p}}_{\Omega}(t)\|_{1}
  =
  1 - \sum_{\bm{x} \in \Omega}{\tilde{p}_{\Omega}(t, \bm{x})}.
\end{equation}
Clearly, the right hand side can be readily computed from the solution of the reduced system~\eqref{eq:fsp_system}.
From this precise error quantification, we have effective iterative method for solving the CME. Choosing an error tolerance $\varepsilon>0$, starting from any initial set $\Omega:=\Omega_{0}$, we solve system~\eqref{eq:fsp_system} and check that the right hand side of~\eqref{eq:fsp_error_bound} is less than $\varepsilon$. If this fails, we add more states to $\Omega_0$ to get a strictly larger set $\Omega_1$ and repeat the procedure until we find an approximation that satisfies our error tolerance.

As the sequence of subsets $\Omega_i$ grows until it eventually covers the whole state space, we might expect that the finite-time solution of~\eqref{eq:fsp_system} will likewise converge to the true solution. This is indeed the case for all models in practice, with only a few theoretical counterexamples in which the Markov chain is explosive~\cite{Munsky2006}. Sufficient conditions for the convergence of the FSP can be checked based on the form of the propensity function~\cite{Gauckler2014, Gupta2014}. In practice, reaction networks tend to have only reactions between two or fewer molecules, with propensity functions in mass-action form~\eqref{eq:propensity_mass_action}, and these are guaranteed to be approximable with the FSP~\cite{Gauckler2014}.

For the rest of the paper, we only concern ourselves with non-explosive SRNs where the FSP converges.
Given such models, any exhaustive sequence of subsets $\{\Omega_j\}$ suffices to guarantee that the FSP solution eventually satisfies any prespecified error tolerance.
However, some choices are more effecient than other, and it is also more advantageous to partition the time interval into smaller timesteps and use a smaller $\Omega_j$ on each step~\cite{Munsky2007}.
We will return to these observations in section~\ref{sec:adaptive_fsp}.

A final point to make about the FSP is that, if the sequence of reduced sets $\Omega_j$ are increasing, that is, $\Omega_j \subset \Omega_{j+1}, j=1,2,\ldots$, the resulting truncation error has been shown to decrease monotonically~\cite[Theorem 1]{Munsky2008} (also see~\cite[Theorem 2.5]{Kuntz2019}). This gives us a natural way to form a hierarchy of reduced models for the CME, and we will return to this in section~\ref{sec:surrogate_cme}. With the computed distribution of the SRN in our hands, we can now match them directly to experimental data using a straightforward likelihood function.
\subsection{Bayesian inference of SRN models from discrete single-cell measurements}
We are interested in inferring the parameters for the reaction network from discrete, single-cell datasets~\cite{Femino1998,Raj2008, Neuert2013a, Munsky2015, Li2019} that consist of several snapshots of many independent cells taken at discrete times $t_1,\ldots, t_T$. The snapshot at time $t_i$ records gene expression in $n_i$ cells, each of which can be collected in the data vector $\bm{c}_{j,i}, \; j = 1,\ldots, n_i$ of molecular populations in cell $j$ at time $t_i$.

Assume that a model class $\mathcal{M} = \left\{M(\theta)\right\}_{\theta \in \Theta}$ of stochastic reaction networks has been chosen to model the data consisting  of a fixed set of reactions with unknown reaction parameters $\theta$.
Let $p(t, \bm{x} | M(\theta))$ denote the entry of the CME solution corresponding to state $\bx$ at time $t$, given by SRN model $M(\theta)$.
The log-likelihood of the dataset $\mathcal{D}$ given $M(\theta)$ is given by
\begin{equation}
  \label{eq:full_logl}
  {L}(\mathcal{D} \vert \theta)
  =
  \sum_{i=1}^{T}
  \sum_{j=1}^{n_i}
  {\log{p(t_i, \bm{c}_{j,i} | M(\theta))}}.
\end{equation}

A common approach to fitting this model is find model parameters $\theta$ that maximize the log-likelihood function. These parameters would be those that best fit the data, but this process does not capture any uncertainty about those parameters. There are potentially many other models that could fit the data approximately as well as the maximum likelihood model. In contrast, with a Bayesian approach we seek to quantify this parameter uncertainty so that we can be confident in the inference results, understand the influence of this uncertainty on future predictions, and design future experiments to reduce the parameter uncertainty.

Bayesian inference is rooted in the Bayesian philosophy of probability in which our uncertainty about the world is modeled using probability distributions~\cite{jaynes2003probability}. As information becomes available, we update these distributions to reflect our new state of understanding about the world. Therefore, for Bayesian inference we begin with a prior distribution, $p \left ( \theta \right )$, that captures our initial beliefs about model parameters. Then after data has been observed, the likelihood of the data given a model class and associated parameters can be found as $p \left (\mathcal{D} \mid \theta \right ) = \exp({L}(\mathcal{D} \vert \theta))$. By applying Bayes' Theorem, we can combine the prior and likelihood information together to construct the posterior distribution on model parameters that reflects our updated beliefs after data has been observed:

\begin{equation}
  \label{eq:Bayes}
  p \left ( \theta \mid \mathcal{D} \right ) = \frac{p \left (\mathcal{D} \mid \theta \right )p \left ( \theta \right )}{p \left ( \mathcal{D} \right )}
\end{equation}

Here, $p \left ( \mathcal{D} \right )$ is a normalization constant known as the model evidence. This is relevant for Bayesian model selection problems but is not required for parameter calibration. When $p \left ( \theta \right )$ is a constant, the parameters that maximize the posterior density are equivalent to the maximum likelihood estimator. Solving the Bayesian inference problem to identify parameters allows us to quantify our uncertainty regarding the accuracy of the parameter fit by sampling from the posterior distribution. However, it can be a computationally challenging problem as discussed in the next subsection.

The Bayesian framework also provides a criteria to select or weigh different model classes as data become available. Suppose, instead of a single model class, we are given $K$ possible network structures that could potentially explain the observations. Let $\mathcal{M}^k = \{M^k(\theta^k)\}_{\theta^k \in \Theta^k}$ denote the $k$-th class, where the parameter domains $\Theta^k$ need not have the same dimensionality. Each model class is associated with a prior weight $P(\mathcal{M}^k)$ that represents the prior level of belief in each class. If $\mathcal{D}$ denotes the dataset as before, we can compute the model evidence of $\mathcal{M}^k$ as
\begin{equation}
  P(\mathcal{D} | \mathcal{M}^k)
  =
  \int_{\theta \in \Theta^k}
  {
  p \left (\mathcal{D} \mid \theta,   \mathcal{M}^k \right )P(\theta \mid \mathcal{M}^k)
  \operatorname{d\theta}
  }.
\end{equation}

\noindent The posterior probability of each model class can then be computed by applying Bayes' Theorem:

\begin{equation}
P \left (\mathcal{M}^k \mid  \mathcal{D}\right)  = \frac{P(\mathcal{D} | \mathcal{M}^k) P(\mathcal{M}^k)}{\sum_{j=1}^K P(\mathcal{D} | \mathcal{M}^j) P(\mathcal{M}^j)}
\end{equation}

These probabilities then reflect the posterior weighting of the different model classes that can be used to make average predictions over the models. Similarly, the Bayes factors
\begin{equation}
  \frac{P \left (\mathcal{M}^k \mid  \mathcal{D}\right)}
  {P \left (\mathcal{M}^j \mid  \mathcal{D}\right)},
  j, k = 1, \ldots, K, \; j \neq k
\end{equation}
allow us to rank and compare the models based on how strongly they are supported by the current data for Bayesian model selection. The drawback of Bayesian model selection in the context of stochastic gene expression is the computational cost of computing the model evidences.
We will return to this issue in section~\ref{sec:multifi_stmcmc} where we show how the Mutlfidelity ST-MCMC framework provides an efficient way to estimate these evidences.

\subsection{Markov Chain Monte Carlo samplers}
\label{sec:MCMCTheory}
Markov Chain Monte Carlo algorithms are widely used for sampling the posterior distribution of a Bayesian inference problem. These methods design a Markov chain whose stationary distribution is the target posterior distribution, $p \left ( \theta \mid \mathcal{D} \right )$. Therefore, by simulating the evolution of samples, $\theta$, according to the Markov chain, samples are asymptotically drawn from the posterior distribution. However, unlike Monte Carlo sampling, these samples are correlated. A common MCMC method is the Metropolis-Hastings algorithm. The algorithm begins by initializing the parameter state to some $\theta_0$. Then at a step $i+1$ in the evolution, a candidate sample $\theta^\prime$ is drawn according to a proposal distribution $Q \left (\theta^\prime \mid \theta_{i}\right )$. This candidate is then accepted or rejected with probability $\alpha \left (\theta^\prime \mid \theta_i \right )$  given by

\begin{equation}
\alpha \left ( \theta^\prime \mid \theta_i \right ) = \min \left ( 1, \frac{p \left ( \theta^\prime \mid \mathcal{D} \right ) Q \left ( \theta_{i} \mid \theta^\prime \right ) }{p \left ( \theta_i \mid \mathcal{D} \right ) Q \left (\theta^\prime \mid  \theta_{i} \right ) } \right ) = \min \left ( 1, \frac{p \left (\mathcal{D} \mid \theta^\prime \right )p \left ( \theta^\prime \right ) Q \left ( \theta_{i} \mid \theta^\prime \right ) }{p \left (\mathcal{D} \mid \theta_i \right )p \left ( \theta_i \right ) Q \left (\theta^\prime \mid  \theta_{i} \right ) } \right )
\end{equation}

\noindent This acceptance probability is independent of the normalization constant $p \left ( \mathcal{D} \right )$ and therefore computationally tractable. If the candidate is accepted $\theta_{i+1} = \theta^\prime$; otherwise, $\theta_{i+1} = \theta_i$. This algorithm iterates until a sufficient number of posterior samples have been generated to accurately represent the posterior.

A common method to judge whether sufficient posterior samples have been generated using the MCMC sampler is the notion of effective sample size (ESS). The ESS of a $N$ sample correlated population $\theta_{i=1...N}$ corresponds to the number of independent samples which would estimate a quantity of interest, $\hat{q} = E \left [ q \left (  \theta \right ) \right ]$, with the same variance as the correlated sample population. Therefore, the ESS of a quantify of interest can be computed as

 \begin{equation}
 N_{ESS} = \frac{Var \left [q \left (  \theta \right )\right  ]}{Var \left [ \hat{q} \right  ]}
 \end{equation}

 Based upon the Markov chain Central Limit Theorem, for correlated samples generated according to the stationary distribution of a Markov process, the ESS can be approximated as

  \begin{equation}
 N_{ESS} = \frac{N}{1 + 2 \sum_{i=1}^N \rho_i \left ( q\left (\theta_{1 \dots N} \right ) \right )}
 \end{equation}

\noindent where $\rho_i$ is the $i$th lag autocorrelation of the quantify of interest whose evolution is defined by the Markov chain. Therefore, when designing a Metropolis-Hastings proposal distribution $Q \left (\theta^\prime \mid \theta_{i}\right )$ we want minimize the correlation and thus maximize $\frac{N_{ESS} }{N}$ to attain the highest possible sampling efficiency. However, even very effective samplers often need tens of thousands of sequential model evaluations to generate enough effective samples, making this form of MCMC challenging for computationally expensive models.

\subsection{Sequential Monte Carlo samplers}
\label{sec:SMCTheory}
To overcome many of the challenges associated with a standard Metropolis-Hastings based MCMC method, parallel methods, like Sequential Monte Carlo (SMC), have been introduced to better leverage high performance computing resources. SMC methods for Bayesian inference transport a sample population, initially distributed so that it can approximate expectations with respect to the prior, to one which can approximate posterior expectations~\cite{Moral2006, ChingChen2007, Kantas2014, catanach2018bayesian}. Typically, this means a population of samples initially distributed according to the prior being transformed into a population of samples approximately distributed according to the posterior. For Sequential Tempered MCMC (ST-MCMC)~\cite{catanach2018bayesian}, we break down the inference problem into a series of annealing levels $i$ defined by an annealing factor $\beta_i \in \left [0, 1 \right ]$. Each level defines an intermediate distribution, $\pi_{\beta_i}\left (\theta \right )$, which we would like to use to generate samples and to compute expectations. These intermediate distributions take the form of

\begin{equation}
\pi_{\beta_i}\left (\theta \right ) = p \left (\theta \mid \mathcal{D}, \beta_i \right) = \frac{p \left (\mathcal{D} \mid \theta \right )^{\beta_i} p \left ( \theta \right )}{\int p \left (\mathcal{D} \mid \theta \right )^{\beta_i} p \left ( \theta \right ) d\theta}
\end{equation}

This annealing approach is common to many SMC methods used for Bayesian inference and can be thought of as gradually integrating the influence of the data into the solution. To simplify the problem of transporting samples from the prior, $\pi_0\left (\theta \right )$, to posterior, $\pi_1\left (\theta \right )$, we transport samples sequentially through each level in the sequence, i.e. $\pi_{\beta_i}\left (\theta \right )$ to $\pi_{\beta_{i+1}}\left (\theta \right )$ with $\beta_{i+1} > \beta_{i}$. Because we can control the size of the jump $\Delta \beta = \beta_{i+1}-\beta_i$, we can ensure that this change is not too drastic as to cause poor approximation of the true distribution i.e. too drastic a decrease in the ESS. Transporting samples is done in three steps:

\begin{enumerate}
	\item Re-weight the previous sample population, distributed according to $\pi_{\beta_i}\left (\theta \right )$, with unnormalized weights $w_i = p \left (\mathcal{D} \mid \theta_i \right )^{\Delta \beta}$ to reflect expectations with respect to the new distribution $\pi_{\beta_{i+1}}\left (\theta \right )$.
	\item Re-sample the population according to the weights so that the samples now reflect $\pi_{\beta_{i+1}}\left (\theta \right )$.
	\item Seed a Markov chain starting at each sample, and then use MCMC to explore $\pi_{\beta_{i+1}}\left (\theta \right )$.
\end{enumerate}

The MCMC step is essential to ensure the sample population does not degenerate since the re-weighting and re-sampling steps reduce the ESS of the population. MCMC increases the ESS because it decorrelates the seeds and explores the target distribution, causing samples to better reflect it. Typically, $\Delta \beta$ is chosen adaptively to not decrease the ESS too much during the update. This is achieved by finding a $\Delta \beta$ such that the coefficient of variation (COV) of the sample weights equals a target $\kappa$. The COV approximates the ESS by $N_{ESS} \approx \frac{N}{1 + \kappa^2}$. Therefore, we find a $\Delta \beta > 0$ that solves the equation

\begin{equation}
\kappa = \frac{\sqrt{\frac{1}{N} \sum_{i=1}^N \left ( w_i \left ( \Delta \beta \right ) - \hat{w}\left ( \Delta \beta \right ) \right)^2}}{\hat{w}\left ( \Delta \beta \right )}
\label{eq:betatune}
\end{equation}

\noindent Here, $w_i\left ( \Delta \beta \right ) = p \left (\mathcal{D} \mid \theta_i \right )^{\Delta \beta}$ and $\hat{w}\left ( \Delta \beta \right ) = \frac{1}{N} \sum_{i=1}^N w_i\left ( \Delta \beta \right )$. Typically, we choose $\kappa = 1$, which corresponds to a target ESS of $N/2$. With this method for finding $\Delta \beta$, we then sequentially move through all the adaptively tuned annealing levels until we reach the final posterior reflected by $\pi_1\left (\theta \right )$. For more details about this algorithm, see~\cite{catanach2018bayesian}.

\section{Multifidelity ST-MCMC}
\label{sec:multifi_stmcmc}
For expensive models, ST-MCMC and similar SMC based methods may still be computationally prohibitive. One approach to overcome this computational burden is to utilize a multifidelity model hierarchy that can speed up sampling. The key idea is that for early levels of the ST-MCMC algorithm a low fidelity but computationally cheap model may be sufficiently informative to guide the samples towards the ultimate posterior distribution. This is because at early levels, the annealing factor causes the contribution of the likelihood to be damped, so perturbations in the likelihood caused by the decrease in model fidelity are less important. Intuitively, a lower fidelity model may be useful when the bias it introduces in the likelihood function is less than the variance of the likelihood at the annealing level. We consider different strategies for rigorously defining this intuition in the rest of this section.

We extend ST-MCMC to multifidelity ST-MCMC by defining intermediate levels both in terms of their annealing factor $\beta$ and the choice of model fidelity $M$, where we consider a hierarchy of models $\mathcal{M} = \{ M_j : j = 1 \dots K\}$ with increasing fidelity and computational cost. This algorithm is described in Algorithm \ref{MFSTMCMC}. The key challenge is determining the best strategy for choosing $\beta_l$ and $m_l$ and Step 2, since the rest of the algorithm proceeds like standard ST-MCMC.

\begin{algorithm}
\SetKwInOut{Input}{Inputs}\SetKwInOut{Output}{Output}\SetKwInOut{Initialization}{Initialize}
\Input{ Prior distribution $p \left ( \theta \right ) $\\
Model fidelity hierarchy $\mathcal{M} = \{ M_j : j = 1 \dots K\}$\\
Likelihood function $p \left ( \mathcal{D} \mid \theta,  M_j \right )$ \\
 Number of samples $N$}
\Output{Posterior samples $\theta_{1 \dots N}$}
\BlankLine
\Begin{
\Initialization{Set level counter $l=0$ \\
 Set annealing factor $\beta_0 = 0$\\
 Set model level $m_0 = 1$ \\
 Define the first intermediate distribution $\pi_0 \left ( \theta \right ) = p \left ( \theta \right )$ \\
 Draw initial samples $\theta_{1 \dots N}^0 \sim \pi_0 \left ( \theta \right ) $ }
\While{$\pi_{\beta_l} \left ( \theta \mid M_{m_l} \right ) \neq p \left ( \theta \mid \mathcal{D}, M_K \right )$}{
1) Increment level counter $l = l + 1$\\
2) Choose the next $\beta_l$ and $m_l$ based on the previous level sample population $\theta_{1 \dots N}^{l-1}$\\
3) Define the next intermediate distribution $\pi_{\beta_l} \left ( \theta \mid M_{m_l} \right ) \propto p \left ( \mathcal{D} \mid \theta,  M_{m_l} \right )^{\beta_l} p \left ( \theta \right )$  \\
4) Compute the unnormalized importance weights for the population as $w_i = \frac{\pi_{\beta_l} \left ( \theta_i^{l-1} \mid M_{m_l} \right )}{\pi_{\beta_{l-1}} \left ( \theta_i^{l-1} \mid M_{m_{l-1}} \right )}$\\
5) Resample the population according to the normalized importance weights to initialize $\theta_{1 \dots N}^{l}$\\
6) Evolve the samples $\theta_{1 \dots N}^{l}$ using MCMC with stationary distribution $\pi_{\beta_l} \left ( \theta \mid M_{m_l} \right )$
}
\Return $\theta_{1 \dots N} = \theta_{1 \dots N}^{l}$
}
\caption{Multifidelity ST-MCMC \label{MFSTMCMC}}
\end{algorithm}

\subsection{Tempering and bridging using an Effective Sample Size criteria}

One approach to choosing the appropriate annealing factor and model fidelity is a combined likelihood tempering and model bridging scheme discussed in Latz et al.~\cite{Latz2018}. This scheme is based upon the ESS statistic discussed in Section \ref{sec:SMCTheory}. Within Multifidelity ST-MCMC, at every level $l$ of Algorithm \ref{MFSTMCMC}, we choose whether to temper by changing $\beta_l = \beta_{I-1}+\Delta \beta$ or to bridge by changing the model fidelity, $m_l = m_{l-1} + 1$. This choice is made by measuring the ESS of the sample population with respect to the next change in the model fidelity by computing the unnormalized weights as if we were to bridge:

\begin{equation}
w_i = \frac{\pi_{\beta_l-1} \left ( \theta_i^{l-1} \mid M_{m_{l-1}+1}\right )}{\pi_{\beta_{l-1}} \left ( \theta_i^{l-1} \mid M_{m_{l-1}} \right )} = \frac{p \left ( \mathcal{D} \mid \theta_i^{l-1},  M_{m_{l-1}+1} \right )^{\beta_{l-1}} }{p \left ( \mathcal{D} \mid \theta_i^{l-1},  M_{m_{l-1}} \right )^{\beta_{l-1}} }
\end{equation}

We can then compute the coefficient of variation of the weights, $w_i$, to determine if it exceeds a target $\kappa$. If it does, we choose to bridge to the next model fidelity because the sample population is beginning to degenerate, so it no longer has sufficient ESS to approximate the next level intermediate posterior. If the COV is less than $\kappa$, we choose to keep the current model but instead temper $\beta$. The next beta is chosen using the same strategy as before by solving Equation \ref{eq:betatune}.

\subsection{Information-theoretic criteria for model fidelity adaptation}
\label{sec:it_comp}
We introduce a new criteria for model fidelity selection based on information theory. This criteria is motivated by the fact that the ESS-based strategy described above only decides to change fidelity based upon the next model in the hierarchy. This means the sampler may continue to use a low fidelity model because it still meets the ESS criteria with respect to the next model, even when it is drifting away from the high fidelity posterior.  Instead, we introduce a method that utilizes a limited number of full fidelity model evaluations to help us better decide when to bridge model fidelity. Depending on the computational cost of the full fidelity simulations, the improved bridging strategy and the improved robustness of this method may outweigh the cost of these full fidelity solutions.

Within  multifidelity ST-MCMC, if the algorithm is at annealing level $l$ with annealing factor $\beta_l \in \left [0, 1\right ]$ and has been sampling the intermediate posterior defined by a model, $M_{m_l}$, in a model hierarchy $\mathcal{M} = \{ M_j : j = 1 \dots K\}$, we would like to know whether $M_{m_l}$ still provides meaningful information about the ultimate posterior once we move to level $l+1$ with annealing factor $\beta_{l+1}$. Here we assume the ultimate posterior is $p \left (\theta \mid \mathcal{D}, M_K \right )$, where $M_K$ is the highest fidelity model. Therefore, unlike the previous ESS-based method, we begin by proposing a tempering step under the assumption that the current model fidelity is valid. We find the proposed $\beta_{l+1}$ by solving Equation \ref{eq:betatune}.

If $M_{m_l}$ no longer provides meaningful information at the next level, we use the next highest fidelity model in the algorithm, $M_{m_l+1}$. This criteria can be formulated using a generalization of information theory~\cite{duersch2019generalizing}, where the information gained about the full posterior, $p \left (\theta \mid \mathcal{D}, M_K \right )$, by moving from level $l$ to $l+1$ with model $M_{m_l}$ is:

\begin{equation}
\begin{split}
&\mathcal{I}_{p \left (\theta \mid \mathcal{D}, M_K \right )} \left [ p \left (\theta \mid \mathcal{D}, M_{m_l}, \beta_{l+1} \right ) || p \left (\theta \mid \mathcal{D}, M_{m_l}, \beta_l \right )\right ]\\
&=D_{\text{KL}} \left [p \left (\theta \mid \mathcal{D}, M_K \right ) || p \left (\theta \mid \mathcal{D}, M_{m_l}, \beta_l \right )\right ] -D_{\text{KL}} \left [p \left (\theta \mid \mathcal{D}, M_K \right ) || p \left (\theta \mid \mathcal{D}, M_{m_l}, \beta_{l+1} \right )\right ]\\
&=\int p \left (\theta \mid \mathcal{D}, M_K \right ) \log \frac{p \left (\theta \mid \mathcal{D}, M_{m_l}, \beta_{l+1} \right )}{p \left (\theta \mid \mathcal{D}, M_{m_l}, \beta_l \right )} d\theta
\end{split}
\label{eq:info}
\end{equation}

\noindent If this quantity is positive, then the intermediate posterior defined by $\beta_{l+1}$ and $m_l$ is closer to the ultimate posterior than the previous level, so we choose $m_{l+1} = m_{l}$. However if this is negative, this update is driving the distribution away from the ultimate posterior, so we should use a higher fidelity model for the next update, thus $m_{l+1} = m_{l}+1$.

If we choose to update the model fidelity, we consider two strategies for choosing $\beta_{l+1}$ for the next level. In the first strategy, keeping with the ESS-based tempering and bridging framework from above, is to set $\beta_{l+1} = \beta_{l}$. The second strategy is to tune $\beta_{l+1}$ to try attain an ESS target. The first strategy is often more computationally efficient, but may not be as robust if changing model fidelity introduces significant variations. To tune $\beta_{l+1}$, we first define the importance weight for transitioning from a level defined by $\beta_i$ and $m_{l}$ to a level defined by $\beta_{l+1}$ and $m_{l+1} = m_{l}+1$ as

\begin{equation}
  \label{eq:beta_tuning}
w_i = \frac{\pi_{\beta_{l+1} \left ( \theta_i^{l} \mid M_{m_{l}+1}\right )}}{\pi_{\beta_{l}} \left ( \theta_i^{l} \mid M_{m_{l}} \right )} = \frac{p \left ( \mathcal{D} \mid \theta_i^{l},  M_{m_{l}+1} \right )^{\beta_{l+1}} }{p \left ( \mathcal{D} \mid \theta_i^{l},  M_{m_{l}} \right )^{\beta_{l}} }
\end{equation}

Using the same approach as before, we can then tune $\beta_{l+1}$ to meet some ESS target based upon the COV of the weights. However, unlike in previous problems, this might not be achievable.  If Equation \eqref{eq:betatune} has a solution, we chose the largest $\beta_{l+1}$ such that the COV target is met. If Equation \eqref{eq:betatune} does not have a solution, we find the $\beta_{l+1}$ that minimizes the COV and thus maximizes the ESS.

\subsection{Computing the information-theoretic criteria}

Since computing the information in Equation \ref{eq:info} requires marginalizing over the posterior, it can be challenging. However, this computation can be approximated using the samples from ST-MCMC. The first step is to recognize the connection between computing this criteria and estimating the model evidence:

\begin{equation}
\begin{split}
&\mathcal{I}_{p \left (\theta \mid \mathcal{D}, M_K \right )} \left [ p \left (\theta \mid \mathcal{D}, M_{m_l}, \beta_{l+1} \right ) || p \left (\theta \mid \mathcal{D}, M_{m_l}, \beta_l \right )\right ]\\
&=\int p \left (\theta \mid \mathcal{D}, M_K \right )\log \frac{p \left (\theta \mid \mathcal{D}, M_{m_l}, \beta_{l+1} \right )}{p \left (\theta \mid \mathcal{D}, M_{m_l}, \beta_l \right )} d\theta\\
&=\int p \left (\theta \mid \mathcal{D}, M_K \right )\log \frac{p \left (\mathcal{D} \mid \theta, M_{m_l} \right )^{\beta_{l+1}} p\left (\theta \right )}{p \left (\mathcal{D} \mid M_{m_l}, \beta_{l+1} \right )} \frac{p \left (\mathcal{D} \mid M_{m_l}, \beta_{l} \right )}{p \left (\mathcal{D} \mid \theta, M_{m_l} \right )^{\beta_l} p\left (\theta \right )} d\theta\\
&= \int p \left (\theta \mid \mathcal{D}, M_K \right )\log p \left (\mathcal{D} \mid \theta, M_{m_l} \right )^{\Delta \beta} \frac{p \left (\mathcal{D} \mid M_{m_l}, \beta_{l} \right )}{p \left (\mathcal{D} \mid M_{m_l}, \beta_{l+1} \right )} d\theta
\end{split}
\end{equation}

\noindent Here, $p \left (\mathcal{D} \mid M, \beta \right )$ is the model evidence, i.e. the normalization, for the likelihood defined by the model $M$ with an annealing factor $\beta$:

\begin{equation}
p \left (\mathcal{D} \mid M, \beta \right ) = \int p \left (\mathcal{D} \mid \theta, M \right )^{\beta} p \left ( \theta \right ) d\theta
\end{equation}

\noindent By noting the relationship to model evidence, the ratio of the evidences can be expressed as:

\begin{equation}
\begin{split}
\frac{p \left (\mathcal{D} \mid M_{m_l}, \beta_{l} \right )}{p \left (\mathcal{D} \mid M_{m_l}, \beta_{l+1} \right )} &= \frac{\int p \left (\mathcal{D} \mid \theta, M_{m_l} \right )^{\beta_l} p \left ( \theta \right ) d\theta}{\int p \left (\mathcal{D} \mid \theta, M_{m_l} \right )^{\beta_{l+1}} p \left ( \theta \right ) d\theta}\\
&= \frac{1}{\int p \left (\mathcal{D} \mid \theta, M_{m_l} \right )^{\Delta \beta} \frac{p \left (\mathcal{D} \mid \theta, M_{m_l} \right )^{\beta_l} p \left ( \theta \right )}{\int p \left (\mathcal{D} \mid \theta, M_{m_l} \right )^{\beta_l} p \left ( \theta \right ) d\theta}  d\theta}\\
&= \frac{1}{\int p \left (\mathcal{D} \mid \theta, M_{m_l} \right )^{\Delta \beta} p \left (\theta \mid \mathcal{D}, M_{m_l}, \beta_l \right )  d\theta}\\
&= \frac{1}{\text{E}_{\theta \sim p \left (\theta \mid \mathcal{D}, M_{m_l}, \beta_l \right )} \left [ p \left (\mathcal{D} \mid \theta, M_{m_l} \right )^{\Delta \beta} \right ]}
\end{split}
\end{equation}

\noindent Therefore,

\begin{equation}
\begin{split}
&\int p \left (\theta \mid \mathcal{D}, M_K \right )\log p \left (\mathcal{D} \mid \theta, M_{m_l} \right )^{\Delta \beta} \frac{p \left (\mathcal{D} \mid M_{m_l}, \beta_{l} \right )}{p \left (\mathcal{D} \mid M_{m_l}, \beta_{l+1} \right )} d\theta\\
&= \int p \left (\theta \mid \mathcal{D}, M_K \right )\log \frac{p \left (\mathcal{D} \mid \theta, M_{m_l} \right )^{\Delta \beta}}{\text{E}_{\theta \sim p \left (\theta \mid \mathcal{D}, M_{m_l}, \beta_l \right )} \left [ p \left (\mathcal{D} \mid \theta, M_{m_l} \right )^{\Delta \beta} \right ]} d\theta
\end{split}
\end{equation}

Since we cannot yet sample $p \left (\theta \mid \mathcal{D}, M_K \right )$ we use importance sampling to express this integral in terms of the level $l$ distribution, which we have samples for:

\begin{equation}
\begin{split}
&\int p \left (\theta \mid \mathcal{D}, M_K \right )\log \frac{p \left (\mathcal{D} \mid \theta, M_{m_l} \right )^{\Delta \beta}}{\text{E}_{\theta \sim p \left (\theta \mid \mathcal{D}, M_{m_l}, \beta_l \right )} \left [ p \left (\mathcal{D} \mid \theta, M_{m_l} \right )^{\Delta \beta} \right ]} d\theta\\
&= \int p \left (\theta \mid \mathcal{D}, M_{m_l}, \beta_l \right ) \frac{p \left (\theta \mid \mathcal{D}, M_K \right )}{p \left (\theta \mid \mathcal{D}, M_{m_l}, \beta_l \right )}\log \frac{p \left (\mathcal{D} \mid \theta, M_{m_l} \right )^{\Delta \beta}}{\text{E}_{\theta \sim p \left (\theta \mid \mathcal{D}, M_{m_l}, \beta_l \right )} \left [ p \left (\mathcal{D} \mid \theta, M_{m_l} \right )^{\Delta \beta} \right ]} d\theta\\
&\propto \int p \left (\theta \mid \mathcal{D}, M_{m_l}, \beta_l \right ) \frac{p \left (\mathcal{D} \mid \theta, M_k \right )}{p \left (\mathcal{D} \mid \theta, M_{m_l} \right )^{\beta_l}}\log \frac{p \left (\mathcal{D} \mid \theta, M_{m_l} \right )^{\Delta \beta}}{\text{E}_{\theta \sim p \left (\theta \mid \mathcal{D}, M_{m_l}, \beta_l \right )} \left [ p \left (\mathcal{D} \mid \theta, M_{m_l} \right )^{\Delta \beta} \right ]} d\theta\\
&=\mathcal{I}_{criteria}
\end{split}
\end{equation}

\noindent This integral only needs to be known up to a constant of proportionality since we only need to assess whether it is positive. We can then express it in terms of expectations as:

\begin{equation}
\begin{split}
\mathcal{I}_{criteria} = & \text{E}_{\theta \sim p \left (\theta \mid \mathcal{D}, M_{m_l}, \beta_l \right )} \left [ \frac{p \left (\mathcal{D} \mid \theta, M_K \right )}{p \left (\mathcal{D} \mid \theta, M_{m_l} \right )^{\beta_l}} \log p \left (\mathcal{D} \mid \theta, M_{m_l} \right )^{\Delta \beta} \right ] -\\
&\text{E}_{\theta \sim p \left (\theta \mid \mathcal{D}, M_{m_l}, \beta_l \right )} \left [ \frac{p \left (\mathcal{D} \mid \theta, M_K \right )}{p \left (\mathcal{D} \mid \theta, M_{m_l}\right )^{\beta_l}} \right ] \log \text{E}_{\theta \sim p \left (\theta \mid \mathcal{D}, M_{m_l}, \beta_l \right )} \left [ p \left (\mathcal{D} \mid \theta, M_{m_l} \right )^{\Delta \beta} \right ]
\end{split}
\end{equation}

\noindent We can now estimate whether $\mathcal{I}_{criteria}$ is positive or negative to determine if information is gained or lost by this next update. To approximate these expectations we use the $N$ ST-MCMC samples at level $l$, where ${\theta^l_{i}: i = 1 \dots N}$, which are approximately distributed according to $ p \left (\theta \mid \mathcal{D}, M_{m_l}, \beta_l \right )$. We also use the evaluation of the full fidelity model likelihood at these points:

\begin{equation}
\begin{split}
\text{E}_{\theta \sim p \left (\theta \mid \mathcal{D}, M_{m_l}, \beta_l \right )} &\left [ \frac{p \left (\mathcal{D} \mid \theta, M_K \right )}{p \left (\mathcal{D} \mid \theta, M_{m_l} \right )^{\beta_l}} \log p \left (\mathcal{D} \mid \theta, M_{m_l} \right )^{\Delta \beta} \right ]\\
& \approx \sum_{l=1}^{N} \frac{p \left (\mathcal{D} \mid \theta^l_i, M_K \right )}{p \left (\mathcal{D} \mid \theta^l_i, M_{m_l} \right )^{\beta_l}} \log p \left (\mathcal{D} \mid \theta^l_i, M_{m_l} \right )^{\Delta \beta}
\end{split}
\end{equation}

\begin{equation}
\text{E}_{\theta \sim p \left (\theta \mid \mathcal{D}, M_{m_l}, \beta_l \right )} \left [ \frac{p \left (\mathcal{D} \mid \theta, M_K \right )}{p \left (\mathcal{D} \mid \theta, M_{m_l} \right )^{\beta_l}} \right ] \approx \sum_{l=1}^{N} \frac{p \left (\mathcal{D} \mid \theta^l_i, M_K \right )}{p \left (\mathcal{D} \mid \theta^l_i, M_{m_l} \right )^{\beta_l}}
\end{equation}

\begin{equation}
\text{E}_{\theta \sim p \left (\theta \mid \mathcal{D}, M_{m_l}, \beta_l \right )} \left [ p \left (\mathcal{D} \mid \theta, M_{m_l} \right )^{\Delta \beta} \right ] \approx \sum_{l=1}^{N} p \left (\mathcal{D} \mid \theta^l_i, M_{m_l} \right )^{\Delta \beta}
\end{equation}

\subsection{Multifidelity ST-MCMC and Bayesian model selection}

SMC and ST-MCMC methods not only enable robust solutions of Bayesian inference problems for parameter calibration, but also enable Bayesian model selection by providing asymptotically unbiased estimates of the model evidence. Model evidence estimates are generally highly computationally expensive since they require estimating the normalization constant,

\begin{equation}
p \left ( \mathcal{D} \mid M \right ) = \int p \left ( \mathcal{D} \mid \theta, M \right ) p \left ( \theta \mid M\right ) d\theta
,
\end{equation}
which consists of marginalizing the likelihood over the prior distribution. If the high probability content of the prior differs significantly from the most likely parameters according to the likelihood, it is difficult to estimate this integral using Monte Carlo samples.  Instead, SMC type methods break down this estimate into a series of Monte Carlo approximations over the intermediate distribution levels previously discussed.  As such, a hierarchy of multifidelity models can also be used to accelerate this estimate within the Multifidelity ST-MCMC framework. Using the methods described in~\cite{ChingChen2007, CALDERHEAD20094028} a SMC based sampler, like Multifidelity ST-MCMC, can estimate the model evidence of the highest fidelity model, $M_K$, by estimating the product:

\begin{equation}
p \left ( \mathcal{D} \mid M_K \right ) = \prod_{l = 1}^L \frac{ p \left (\mathcal{D} \mid m_l, \beta_l \right ) }{ p \left (\mathcal{D} \mid m_{l-1}, \beta_{l-1} \right ) } =  \prod_{l = 1}^L c_l
\end{equation}

\noindent where $p \left (\mathcal{D} \mid m, \beta \right ) = \int p \left ( \mathcal{D} \mid \theta, M_{m} \right )^{\beta} p \left ( \theta \right ) d\theta$ and $L$ is the final level of ST-MCMC. The ratio $c_l$ can be written as

\begin{equation}
\begin{split}
c_l &= \frac{\int p \left ( \mathcal{D} \mid \theta, M_{m_l} \right )^{\beta_l} p \left ( \theta \right ) d\theta}{\int p \left ( \mathcal{D} \mid \theta, M_{m_{l-1}} \right )^{\beta_{l-1}} p \left ( \theta \right ) d\theta}\\
&= \int \frac{p \left ( \mathcal{D} \mid \theta, M_{m_l} \right )^{\beta_l}}{p \left ( \mathcal{D} \mid \theta, M_{m_{l-1}} \right )^{\beta_{l-1}}} \frac{p \left ( \mathcal{D} \mid \theta, M_{m_{l-1}} \right )^{\beta_{l-1}} p \left ( \theta \right )}{\int p \left ( \mathcal{D} \mid \theta, M_{m_{l-1}} \right )^{\beta_{l-1}} p \left ( \theta \right ) d\theta} d\theta\\
&=  \text{E}_{\theta \sim p \left (\theta \mid \mathcal{D}, M_{m_{l-1}}, \beta_{l-1} \right )} \left [ \frac{p \left ( \mathcal{D} \mid \theta, M_{m_l} \right )^{\beta_l}}{p \left ( \mathcal{D} \mid \theta, M_{m_{l-1}} \right )^{\beta_{l-1}}} \right ]\\
&\approx \frac{1}{N} \sum_{i=1}^N w_i^l
\end{split}
\end{equation}

\noindent where $w_i^l$ are the unnormalized resampling weights at level $l$ for the sample population $\theta_{i=1 \dots N}^{l-1}$. Therefore, using the weights we already computed as part of Multifidelity ST-MCMC, we are able to compute an estimate of the model evidence.

\section{Multifidelity reduced models of the chemical master equation}
\label{sec:cme_surrogate}
The FSP algorithm introduced in section~\ref{sec:background} is commonly used to compute the likelihood of observed data measurements.
When used within MCMC sampling, the FSP is usually implemented in one of the following two ways.

The first is to fix a single, large, subset of states for all parameter samples~\cite{Gomez-Schiavon2017}.
Since the probability distribution of the CME changes significantly as the MCMC explores the parameter space, it is very difficult to specify a finite state set that accurately captures the significant portion of the probability mass for all times and all parameters. One can end up choosing a static FSP that is either inaccurate or inefficient.
This scenario is similar to when a static discretization scheme (e.g. finite element) is employed in the simulation of parametric partial differential equation models, in which the manually chosen grid size may turn out to be too coarse for some parameter regimes and excessive for others.

This drawback motivates the second approach that instead uses adaptive CME solvers. There have been many adaptive formulations of the FSP~\cite{Munsky2007, Sidje2006, Wolf2010, Sidje2015, Cao2016} in which the state set is iteratively expanded until the approximate solution satisfies a user-specified error tolerance.
However, there could be regions in the parameter space where the adaptive state set has to be expanded to an enormous size to accurately approximate the CME solution.
These `non-physical' parameter combinations usually fit poorly to the data, and the large computational effort for their forward solutions does not provide useful information about the posterior.
On the other hand, since reaction networks usually comprise of nonlinear and unpredictable interactions, it is difficult to know a priori which parameter values would give rise to such difficult (but meaningless) forward solutions of the full-fidelity CME.
Simple techniques to regularize the cost of the forward solutions by restricting either the computational time or the number of time steps may run the risk of mistakenly ignoring genuinely informative parameter candidates whose evaluation just happens to require high computational cost.

The framework of the Multifidelity ST-MCMC sampler allows us to conceive of a compromise. In particular, we recast the static FSP into a surrogate CME whose solution complexity is uniformly bounded across all parameters.
An adaptive FSP method with strict error tolerance is applied only to the surrogate master equations.
This allows us to avoid over committing to parameters that have low posterior probabilities during the early annealing levels, yet still guarantee accurate computation of the likelihood at later sampling stages.

\subsection{Implicitly defined finite state projection for constructing surrogate CME models}
\label{sec:surrogate_cme}
Surrogate models can be derived by adding restrictive assumptions to the physics of the original model.
In particular, consider a hypothetical physical biological surrogate of the original cells modeled by the full CME in which all cellular processes `freeze' when molecular copy numbers reach a certain set of thresholds.
As we increase these thresholds, the surrogate cells behave more freely and closer to the original cells and the master equation describing their behavior becomes closer to the original CME, illustrated in Fig.~\ref{fig:surrogate_cme_ssas}  .

Let $b_1, \ldots, b_N$ be bounds on the copy number of species $1$ through $N$. We define an approximate SRN whose propensities are surrogates of the original SRN propensities and are given by
\begin{equation}
  \label{eq:surrogate_propensity}
  \hat{\alpha}_j(\bm{x})
  =
  \alpha_j(\bm{x})\prod_{i}[x_i \leq b_i],
\end{equation}
where $[E]$ takes value $1$ if expression $E$ is true and zero otherwise.
Since there are no further transitions once the process enters a state that exceeds the bounds, the state space of the surrogate chemical master equation  is effectively reduced to the hyper-rectangle $H(b) = \times_{i=1}^{N}\{0,\ldots, b_i\}$.
Thus, the infinite-dimensional system of differential equations~\eqref{eq:cme_ode_form} is replaced by the finite-dimensional surrogate dynamical system
\begin{equation}
  \label{eq:surrogate_cme}
  M(b):\;
  \frac{d}{dt}{\widehat{\bm{p}}_{H}(t)}
  =
  \widehat{\bm{A}}_{H}
  \widehat{\bm{p}}_{H}(t),
  \;
  \widehat{\bm{p}}_{H}(0) = \widehat{\bm{p}}_0 \vert_{H},
\end{equation}
where the truncated infinitesimal generator $\widehat{A}_{H}$ is defined similar to eq.~\eqref{eq:cme_matrix} but with the exact propensities replaced by the surrogate propensities given in eq.~\eqref{eq:surrogate_propensity}.

We note that eq.~\eqref{eq:surrogate_cme} is equivalent to eq.~\eqref{eq:fsp_system} with $\Omega = H(b)$.
Thus, our surrogate propensities implicitly define a finite state projection of the original CME.
We also note that the surrogate CMEs need not be constrained within a hyper-rectangle as considered here. It may be beneficial to derive a sequence of transformations to the original propensities, with the approximations chosen in such a way that alleviate the computational burden of solving the original model by, e.g., making the lower-fidelity dynamical system less stiff than the high-fidelity one. We leave this more general strategy to future work.

For the present choice hyper-rectangular state spaces, we recall the important result mentioned in section~\ref{sec:fsp_background} that, as the entries of $\bm{b}$ increase monotonically, the state space $H(\bm{b})$ includes more states and the truncation error, measured as the $\ell_1$-distance between $\widehat{\bm{p}}_{H}$ and the true CME solution $\bm{p}(t)$ decreases \emph{monotonically}. This provides us with a straightforward and natural way to form a hierarchy of surrogate models within the Multifidelity ST-MCMC framework.

\begin{figure}
  \centering
  \includegraphics[scale=0.6]{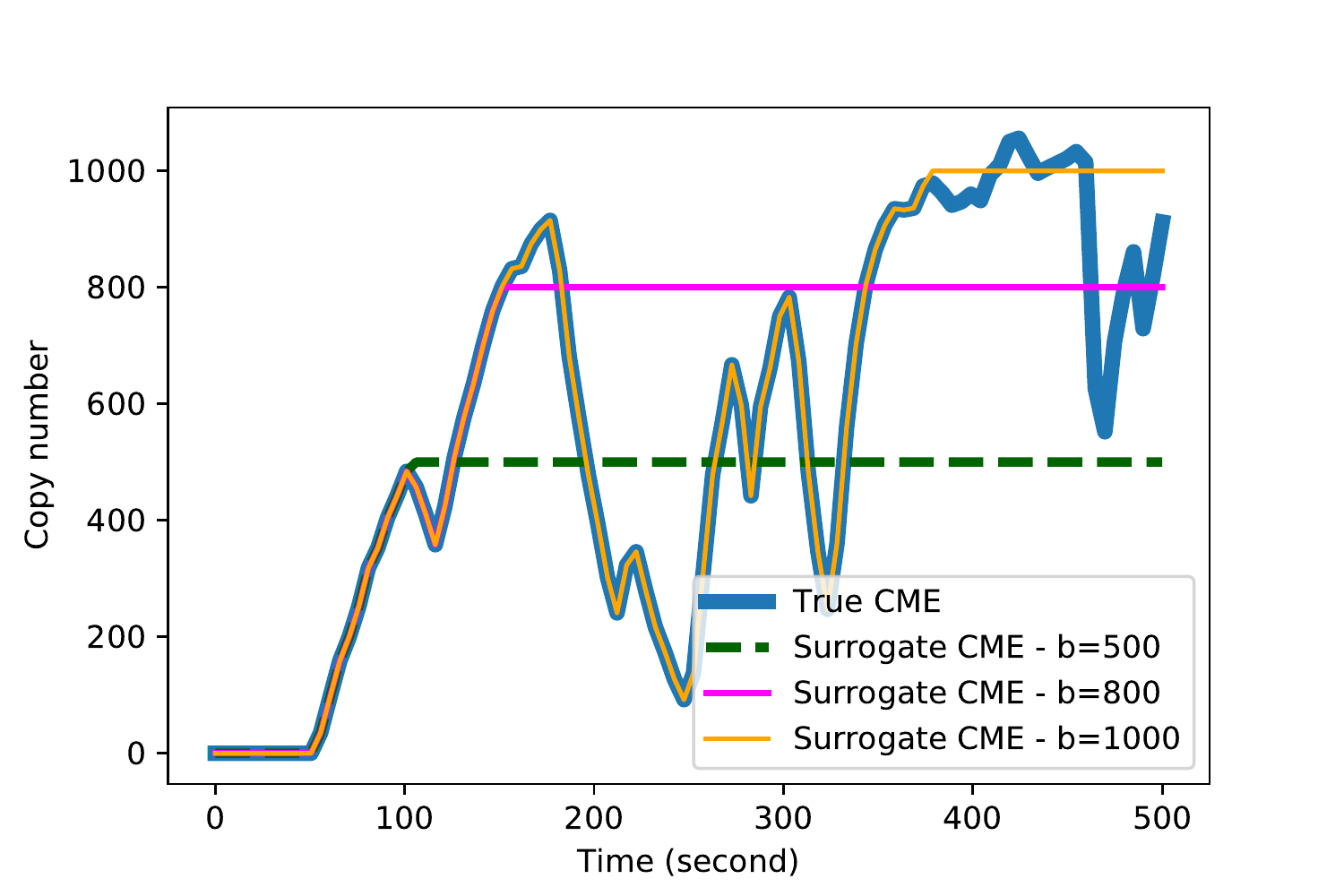}
  \caption{Illustrative realizations of the full and surrogate CMEs for a simple system with mass-action propensity. Given the same random seed, the simulated trajectory of the surrogate CME will be identical to that of the true CME until the state reaches a threshold, $b$, where the surrogate trajectory freezes.
  Increasing the threshold reduces the chance that the surrogate trajectories hit the bounds and consequently more realizations of the true CME are captured by the surrogate model.
  }
  \label{fig:surrogate_cme_ssas}
\end{figure}
\subsection{Using the hierarchy of surrogate CMEs within the ST-MCMC framework}
With the surrogate CME models formulated, a strict hierarchy of surrogate models can be defined by a sequence of bounding vectors $\bm{b}^{(1)} \leq \bm{b}^{(2)} \leq \ldots \leq \bm{b}^{(K)}$ where the ``$\leq$" sign applies element-wise. The corresponding surrogate models $M_{l} := M(\bm{b}^{(l)})$ are then defined as in eq.~\eqref{eq:surrogate_cme}.
As mentioned earlier, the error in the surrogate CMEs decrease monotonically as we increase the bounds. Therefore, $\{M_l\}$ forms a hierarchy in which each level attains more fidelity than its predecessor.

At the $l$-th level, the log-likelihood function in eq.~\eqref{eq:full_logl} is approximated by
\begin{equation}
  \label{eq:surrogate_logl}
  L(\mathcal{D} \vert \theta)
  \approx
  L_{M_l}(\theta)
  =
  \sum_{i=1}^{T}
  {
    \sum_{j=1}^{n_i}
    {
      \log
      p(
        t_i, \min(\bm{c}_{j,i}, \bm{b}^{(l)})
        \vert
        M_{l}(\theta)
      )
    }
  }
  .
\end{equation}
In the surrogate log-likelihood, the data is projected onto the finite state space $H(\bm{b}_{l})$, and the probabilities of the data at different time points  are computed from the surrogate Markov model $M_{l}$. Clearly, as $l$ increases, the surrogate function $L_{M_{l}}(\theta)$ becomes a more accurate approximation to the true log-likelihood $L(\mathcal{D}\vert \theta)$.
In the ideal situation where the hierarchy is allowed to have infinite depth, these surrogates are guaranteed to form a sequence that converges pointwise to the true log-likelihood from below.
This is shown formally in the following proposition.

\begin{proposition}
Let the sequence of bounds $\{\bm{b}^{(l)}\}_{l=1}^{\infty} \subset \mathbb{N}^{N}$, where $\bm{b}^{(l)}:=(b_1^{(l)}, \ldots, b_{N}^{(l)})$ be chosen such that
$\bm{b}^{(l)} \leq \bm{b}^{(l+1)}$ elementwise (i.e., $b_j^{(l)} \leq b_j^{(l+1)})$.
Assume that the continuous-time Markov chain underlying the SRN is non-explosive for all parameters and that the initial distribution of the CME~\eqref{eq:cme_ode_form} has finite support. For each fixed value of the parameter $\theta$, we have the following:
\begin{enumerate}
  \item $L_{M_{l}}(\theta) \rightarrow L(\mathcal{D} \vert \theta)$ as $l \rightarrow \infty$.
  \item There exists a subsequence $l_i$ such that $L_{M_{l_i}} \uparrow L(\mathcal{D} \vert \theta)$.
\end{enumerate}
Here, the log-likelihood function $L(\mathcal{D} \vert \theta)$ is defined as in~\eqref{eq:full_logl} and the surrogate $L_{M_{l_i}}(\theta)$ is defined as in~\eqref{eq:surrogate_logl}.
\end{proposition}
\begin{proof}
  Without loss of generality, we assume that the initial distribution is concentrated at a single state $\bm{x}_0$. Let $R$ be the number of reactions that occur during the finite time interval $[0, t_{T}]$ given the starting state $\bm{x}_0$.
  If there exists $\varepsilon > 0$ for which $\mathbb{P}(R > R_{\varepsilon}) \geq \varepsilon$ for every choice of $R_{\varepsilon}$, then we have $\mathbb{P}(R = \infty) \geq \varepsilon$, violating the assumption of non-explosion.
  Thus, for every $\varepsilon >0$, there exists $R_{\varepsilon}$ such that $\mathbb{P}(R > R_{\varepsilon}) < \varepsilon$. Furthermore, we can find $l_{\varepsilon}$ such that $H(\bm{b}^{(l_{\varepsilon})})$ contains all states that are reachable from $\bm{x}_0$ via $R_{\varepsilon}$ reactions or fewer. Thus, the probability for a sample path within $[0, t_{T}]$ of the surrogate CME to ever exceed $H(b^{(l)})$ is less than $\varepsilon$, and the corresponding solution of the surrogate CME is guaranteed to be less than $\varepsilon$ away (in one-norm) from the true CME solution. This proves (i).

  We have $\min(\bm{c}_{j,i}, \bm{b}^{(l)}) = \bm{c}_{j,i}$ for sufficiently large $l$. Entry-wise, the FSP approximations increase monotonically~\cite[Theorem 2.2]{Munsky2006}, so $p(t_i, \bm{c}_{j,i} | M_{l}(\theta))$ increases monotonically. We can then choose the subsequence $\{l_i\}$ from $\{l\}$ by simply truncating the leading elements until $\bm{b}^{(l)}$ disappears from the $\min(,)$ function in eq.~\eqref{eq:surrogate_logl}. This proves (ii).
\end{proof}

In summary, the FSP scheme allows us to define a hierarchy of surrogate master equations that approach the true CME as the surrogate state space enlarges. From this, we can define a sequence of surrogate log-likelihood functions that converge to the true log-likelihood from below.
These surrogates could be used within the Multifidelity ST-MCMC framework introduced in section~\ref{sec:multifi_stmcmc}.
Before we do so, however, we must first ensure that an accurate solution to the system~\eqref{eq:surrogate_cme} could be computed efficiently.
\subsection{Fast and accurate solution of the surrogate master equation}
\label{sec:adaptive_fsp}
Although the surrogate master equation~\eqref{eq:surrogate_cme} is a significant reduction from the infinite-dimensional CME, the number of states included in the truncated state space $H(b)$ still grows as $O(b_1\cdot \ldots \cdot b_{N})$ and the surrogate CME can quickly become expensive as we increase the entries of $b$. However, in practice, the probability mass of the solution vector $\widetilde{\bm{p}}_{H}(t)$ tends to concentrate at a much smaller subset of states. It is therefore advantageous to approximate $\widehat{\bm{p}}_{H}(t)$ with a more compactly supported distribution.
More precisely, if we let $\varepsilon > 0$ be an error tolerance, we can use a distribution $\widetilde{\bm{p}}_{\Omega}$ supported on $\Omega \subset H$ such that $\|\widehat{\bm{p}}_{H} - \widetilde{\bm{p}}_{\Omega}\| \leq \varepsilon$. Here, the FSP error bound~\eqref{eq:full_logl} plays a critical role in choosing the appropriate support set $\Omega$.
We note that this error bound was recently utilized by Fox et al.~\cite{Fox2016} to compute rigorous lower and upper bounds for the true log-likelihood function~\eqref{eq:full_logl}, from which comparison between certain models could be done even at a low-fidelity FSP solution. We do not pursue this direction in the present work.

To efficiently compute the solution of the surrogate CMEs using the principles just mentioned, we employ a new FSP implementation recently developed by Vo and Munsky~\cite{Vo2019parallel}. This solver
divides the time interval of interest $[0, \tf]$ into subintervals $I_j := [t_j, t_{j+1}),\; j=0,\ldots,n_{step}-1$
with $0:=t_0 < t_1 < \ldots < t_{n_{step}}:= \tf$.
On each time subinterval $I_j$, the dense tensor $\widehat{\bm{p}}_{H}(t)$ that is the solution of the surrogate CME~\eqref{eq:surrogate_cme} is approximated by a sparse tensor $\widetilde{\bm{p}}_{\Omega_j}(t)$ supported on $\Omega_j \subset H$, obtained from solving
\begin{equation}
  \label{eq:reduced_surrogate_cme}
  \frac{d}{dt}{\widetilde{\bm{p}}_{\Omega_j}}(t)
  =
  \widetilde{\bm{A}}_{\Omega_j}\widetilde{\bm{p}}_{\Omega_j}(t),
  \;
  t \in [t_j, t_{j+1})
\end{equation}
where
$$
\widetilde{\bm{A}}_{\Omega_j}(\bm{y}, \bm{x}) =
\begin{cases}
  \widehat{\bm{A}}_{H}(\bm{y}, \bm{x}) \; \text{if } \bm{x}, \bm{y} \in \Omega_j \\
  0 \; \text{otherwise}
\end{cases}
.
$$
Clearly, in solving~\eqref{eq:reduced_surrogate_cme}, we only need to keep track of the equations corresponding to states in $\Omega_j$ and that reduces the computational cost significantly.

From the FSP error bound~\eqref{eq:fsp_error_bound}, we derive an error-control criteria of the form
\begin{equation}
  g_j(t) = 1 - \mathds{1}^T\widetilde{\bm{p}}_{\Omega_j}(t)
  \leq
  \frac{t}{\tf}\varepsilon.
\end{equation}
If at some $t \in [t_j, t_{j+1})$ we find that the inequality is not satisfied, more states are added to $\Omega_j$ and the integration starts again from $t_j$ until the criteria is satisfied over the whole interval. The determination of the time steps $t_{j}$ is left to the ODE integrator employed for solving eq.~\eqref{eq:reduced_surrogate_cme}.

The state sets $\Omega_j$ are chosen as integral solutions of a set of inequality constraints. In particular, they have the form
\begin{equation}
  \label{eq:constrained_stateset}
  \Omega_j
  =
  \left\{
  \bm{x} \in H(b)
  \vert
  f_i(\bm{x}) \leq c_i^{(j)}
  \right\},
\end{equation}
where $f_i$ are functions that are chosen a priori, and $c_i>0$ are  positive scalars. To expand $\Omega_j$, we simply increase $c_i^{(j)}$ and run a breadth-first-search routine to explore all reachable states that satisfy the relaxed inequality constraints.

Implementation-wise, the approximate solution $\tilde{\bm{p}}_{\Omega_j}$ is stored in the coordinate format similar to that used for sparse tensors~\cite{Bader2008}.
The list of tensor indices is managed with the Distribute Dictionary data structure in the software package Zoltan~\cite{ZoltanOverviewArticle2002, ZoltanIsorropiaOverview2012}.
We also make use of parallel objects from the PETSc library~\cite{petsc-efficient, petsc-user-ref, petsc-web-page}.
Other implementation details will be communicated elsewhere~\cite{Vo2019parallel}.
It is worth pointing out that these MPI-based libraries allow our implementation to scale into multiple computing nodes, but the scaling efficiency will be limited due to the communication cost inherent in numerical operations such as matrix-vector multiplications. Therefore, simply plugging a parallel forward solution code on an increasing number of nodes into a serial MCMC sampler such as Metropolis-Hastings will have diminishing benefits.
The ST-MCMC, in contrast, allows us to achieve better utilization of the computing resources, since it is embarrassingly parallel. so doubling the number of processors simply enables us to simultaneously sample twice as many parameter samples in about the same computational time.

We also note that using an adaptive solver such as one we present here incurs some numerical error in the surrogate likelihood function. However, we expect this error to be negligible with a conservative choice of error tolerance. In particular, the error threshold $\varepsilon$ is always fixed at $10^{-8}$ in our numerical tests.
In the next section, we will confirm the accuracy and efficiency of our combined multifidelity sampler and adaptive model reduction scheme when applied to two biologically inspired problems and one on a real experimental dataset.
\section{Numerical examples}
\label{sec:numerical}
In the following tests we compare the four variants of the ST-MCMC described above: the Full-fidelity, ESS-Bridge, IT-Bridge, and Tuned IT-Bridge schemes. The full-fidelity scheme is the classic ST-MCMC with every likelihood evaluation using the highest model fidelity. The remaining three schemes are Multifidelity ST-MCMC variants in which the bridging between fidelity levels are determined based on the ESS, the new information theoretic criteria with or without $\beta$-tuning (see eq.~\eqref{eq:beta_tuning} and the preceding discussion in section~\ref{sec:it_comp}). In each of these Multifidelity schemes, the surrogate likelihood is formulated as described in section~\ref{sec:surrogate_cme}.
When all propensities are time-invariant as in the first two examples, the reduced system of ODEs in eq.~\eqref{eq:reduced_surrogate_cme} is solved by computing the action of the matrix exponential operator using the Krylov approximation with Incomplete Orthogonalization Procedure~\cite{Koskela, Vo2017d, GAUDREAULT2018236}, with the Krylov basis size fixed at $30$.
In the case of time-varying propensities in the third example, we use the Four stage third order L-stable Rosenbrock-W scheme~\cite{Rang2005} implemented in the TS module of PETSc~\cite{Abhyankar2018}.
All these ODEs solvers are set with conservative absolute tolerance of $10^{-14}$, and relative tolerance of $10^{-4}$.
\subsection{Parameter inference for repressilator gene circuit}
\def\TetR{\mathrm{TetR}}
\def\lambdacI{\mathrm{\lambda{cI}}}
\def\LacI{\mathrm{LacI}}

\begin{table}
  \centering
  \begin{tabular}{r l l l}
    \toprule
    & reaction & propensity\\
    \midrule
    1. & $\emptyset \rightarrow \TetR$ & $k_0/(1 + a_0[\LacI]^{b_0})$  \\
    2. & $\TetR \rightarrow \emptyset$ & $\gamma_0 [\TetR]$ \\
    3. & $\emptyset \rightarrow \lambdacI$ & $k_1/(1 + a_1[\TetR]^{b_1})$  \\
    4. & $\lambdacI \rightarrow \emptyset$ & $\gamma_1 [\lambdacI]$  \\
    5. & $\emptyset \rightarrow \LacI$ & $k_2/(1 + a_2[\lambdacI]^{b_2})$  \\
    6. & $\LacI \rightarrow \emptyset$ & $\gamma_2 [\LacI]$  \\
    \bottomrule
  \end{tabular}
  \caption{Reactions and propensities in the repressilator model. ([X] is the number of copies of the species X.)}
  \label{table:repressilator_reactions}
\end{table}

We first consider a three-species model inspired by the well-known repressilator gene circuit~\cite{Elowitz2000}. This model consists of three species, $\TetR$, $\lambdacI$ and $\LacI$, which constitute a negative feedback network (Table~\ref{table:repressilator_reactions}).
We simulate a dataset that consists of five measurement times $2, 4, 6, 8$, and $10$ minutes, with $1000$ cells measured at each time point. These numbers of single-cell measurements are typical of smFISH experiments~\cite{Raj2008, Kalb2019}.
We assume that all cells start at the state $\bm{x}_0 = (\TetR, \lambdacI, \LacI) = \left(0, 0, 0\right)$, so that at the initial time where there are no gene products.

The hierarchy of surrogate CMEs (cf.~\eqref{eq:surrogate_cme}) is defined by the bounds
$$
b^{(l)}
=
\begin{bmatrix}
  b_{\TetR}^{(l)} \\
  \\
  b_{\lambdacI}^{(l)} \\
  \\
  b_{\LacI}^{(l)}
\end{bmatrix}
=
\begin{bmatrix}
    \floor*{c_1 + (l-1)\frac{d_1 - c_1}{L_{\max} + 1}} \\
    \\
    \floor*{c_2 + (l-1)\frac{d_2 - c_2}{L_{\max} + 1}} \\
    \\
    \floor*{c_3 + (l-1)\frac{d_3 - c_3}{L_{\max} + 1}}
\end{bmatrix}
$$
where $(c_1, c_2, c_3) = (20, 40, 40)$ and $(d_1, d_2, d_3) = (50, 100, 100)$, with $L_{\max} = 10$. Therefore, the multifidelity ST-MCMC will transit through ten levels, with the highest-fidelity model having a state space of size $51\times 101 \times 101$.

We conduct parameter inference in $\log_{10}$-transformed space. The prior for the parameters are chosen to be a multivariate normal distribution (in $\log_{10}$ space) with a diagonal covariance matrix (see table~\ref{table:repressilator_posterior}). We ran both the ST-MCMC with the highest-fidelity surrogate CME and the multifidelity ST-MCMC on $29$ nodes, with $36$ cores per node with $1044$ parallel chains. For each level, samples were evolved using Metropolis-Hastings MCMC until a correlation target of $0.6$ was reached. The proposal distribution was adaptively tuned as part of the algorithm. Details on the ST-MCMC sampler and its tuning can be found in~\cite{catanach2018bayesian}.
Fig.~\ref{fig:repressilator_performance} shows the time taken for each sampling scheme to reach a certain annealing level, with the multifidelity schemes with our proposed IT-based criteria outperforming the state-of-the-art fixed-fidelity ST-MCMC and ESS Bridging schemes.
Specifically, while the full-fidelity ST-MCMC took over $25$ hours to finish, the Multifidelity ST-MCMC with ESS, Information Theoretic, and Tuned Information Theoretic Bridging took respectively $7.2$, $4.1$ and $5.3$ hours, resulting in speedup factors of about $3.5$, $6.2$ and $4.8$. The novel Information Theoretic (IT) schemes are clearly faster than the ESS-based scheme in this example, with the untuned IT scheme almost twice as fast as the ESS-based scheme. We observe that for the early levels of the algorithms, when $\beta$ is small, the lowest fidelity model is sufficiently informative. Further, at these early levels the ESS-based scheme is slightly faster than the others since it does not require any full model evaluations. However, after $\beta$ gets larger, the IT-based methods start to outperform the ESS-based method since they use the full model evaluations to judge that they do not need to bridge to the higher fidelity models as quickly as the ESS-based scheme does.

Although the prior assigns a probability density of only about $8.766\times 10^{-20}$ to the true parameter vector, all samplers were able to bring the particles close to the true parameters (Fig.~\ref{fig:repressilator_sampling_history}).
There is no notable difference in the shapes of the posterior distributions constructed from the samples of these two schemes (Fig.\ref{fig:repressilator_posteriors} and table~\ref{table:repressilator_posterior}).

\begin{table}
  \centering
    \begin{tabular}{l|r|r|rrrr} 
\toprule\makecell[c]{Parameter}& True & \makecell[c]{Prior} & \multicolumn{4}{c}{Posterior} \\& & & Full-fidelity  & ESS-Bridge  & IT-Bridge  & Tuned IT-Bridge \\ 
\hline$\log_{10}(k_0)$ &  1.00 & $ 10.00 \pm 0.3 $ & $ \text{1.00} \pm 0.01$  & $ \text{1.00} \pm 0.02$  & $ \text{0.99} \pm 0.01$  & $ \text{1.00} \pm 0.01$ \\ 
$\log_{10}(\gamma_0)$ & -2.00 & $ 0.10 \pm 0.3 $ & $ \text{-1.98} \pm 0.07$  & $ \text{-1.98} \pm 0.08$  & $ \text{-1.98} \pm 0.07$  & $ \text{-1.98} \pm 0.07$ \\ 
$\log_{10}(a_0)$ & -1.00 & $ 0.10 \pm 0.3 $ & $ \text{-1.05} \pm 0.06$  & $ \text{-1.05} \pm 0.07$  & $ \text{-1.07} \pm 0.06$  & $ \text{-1.06} \pm 0.06$ \\ 
$\log_{10}(b_0)$ &  0.30 & $ 0.10 \pm 0.3 $ & $ \text{0.31} \pm 0.01$  & $ \text{0.31} \pm 0.01$  & $ \text{0.31} \pm 0.01$  & $ \text{0.31} \pm 0.01$ \\ 
$\log_{10}(k_1)$ &  0.88 & $ 10.00 \pm 0.3 $ & $ \text{0.87} \pm 0.00$  & $ \text{0.87} \pm 0.01$  & $ \text{0.87} \pm 0.00$  & $ \text{0.87} \pm 0.00$ \\ 
$\log_{10}(\gamma_1)$ & -1.70 & $ 0.10 \pm 0.3 $ & $ \text{-1.71} \pm 0.05$  & $ \text{-1.73} \pm 0.06$  & $ \text{-1.73} \pm 0.05$  & $ \text{-1.71} \pm 0.05$ \\ 
$\log_{10}(a_1)$ & -2.00 & $ 0.10 \pm 0.3 $ & $ \text{-1.98} \pm 0.05$  & $ \text{-1.98} \pm 0.06$  & $ \text{-1.99} \pm 0.05$  & $ \text{-1.98} \pm 0.05$ \\ 
$\log_{10}(b_1)$ &  0.40 & $ 0.10 \pm 0.3 $ & $ \text{0.40} \pm 0.01$  & $ \text{0.40} \pm 0.01$  & $ \text{0.40} \pm 0.01$  & $ \text{0.40} \pm 0.01$ \\ 
$\log_{10}(k_2)$ &  1.00 & $ 10.00 \pm 0.3 $ & $ \text{0.98} \pm 0.01$  & $ \text{0.99} \pm 0.01$  & $ \text{0.98} \pm 0.01$  & $ \text{0.98} \pm 0.01$ \\ 
$\log_{10}(\gamma_2)$ & -1.30 & $ 0.10 \pm 0.3 $ & $ \text{-1.34} \pm 0.03$  & $ \text{-1.34} \pm 0.04$  & $ \text{-1.34} \pm 0.03$  & $ \text{-1.34} \pm 0.03$ \\ 
$\log_{10}(a_2)$ & -1.30 & $ 0.10 \pm 0.3 $ & $ \text{-1.35} \pm 0.06$  & $ \text{-1.34} \pm 0.07$  & $ \text{-1.36} \pm 0.06$  & $ \text{-1.35} \pm 0.06$ \\ 
$\log_{10}(b_2)$ &  0.48 & $ 0.10 \pm 0.3 $ & $ \text{0.48} \pm 0.01$  & $ \text{0.48} \pm 0.01$  & $ \text{0.48} \pm 0.01$  & $ \text{0.48} \pm 0.01$ \\ 
\bottomrule\end{tabular}
  \caption{Model parameters in the repressilator example. The second column presents the parameters of the prior distribution, where we use a Gaussian prior in the $\log_{10}$-transformed parameter space with a diagonal covariance matrix. The last four columns present the posterior mean and standard deviation of model parameters estimated using four methods: fixed-fidelity ST-MCMC (Fixed), Multifidelity ST-MCMC with ESS-Bridging, Multifidelity ST-MCMC with IT-Bridging, and Multifidelity ST-MCMC with $\beta$-tuning and IT-Bridging.}
  \label{table:repressilator_posterior}
\end{table}

\begin{figure}[H]
  \centering
  \includegraphics[scale=0.75]{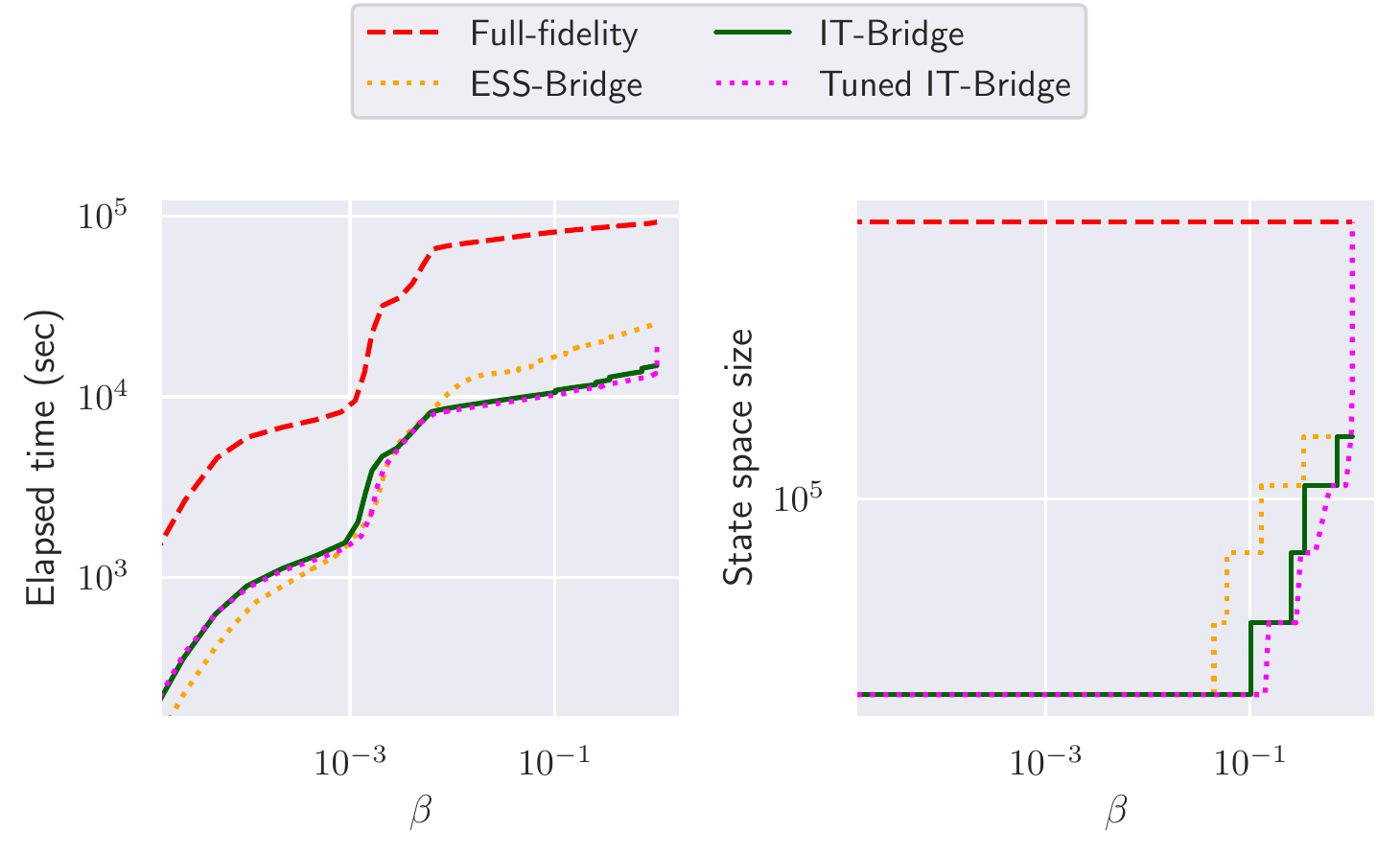}
  \caption{Performance of ST-MCMC samplers on the repressilator example. The horizontal axis represent the annealing factor, i.e. inverse temperature. Significant speed up is observed for the multifidelity schemes.}
  \label{fig:repressilator_performance}
\end{figure}

\begin{figure}
\centering
\includegraphics[scale=0.85]{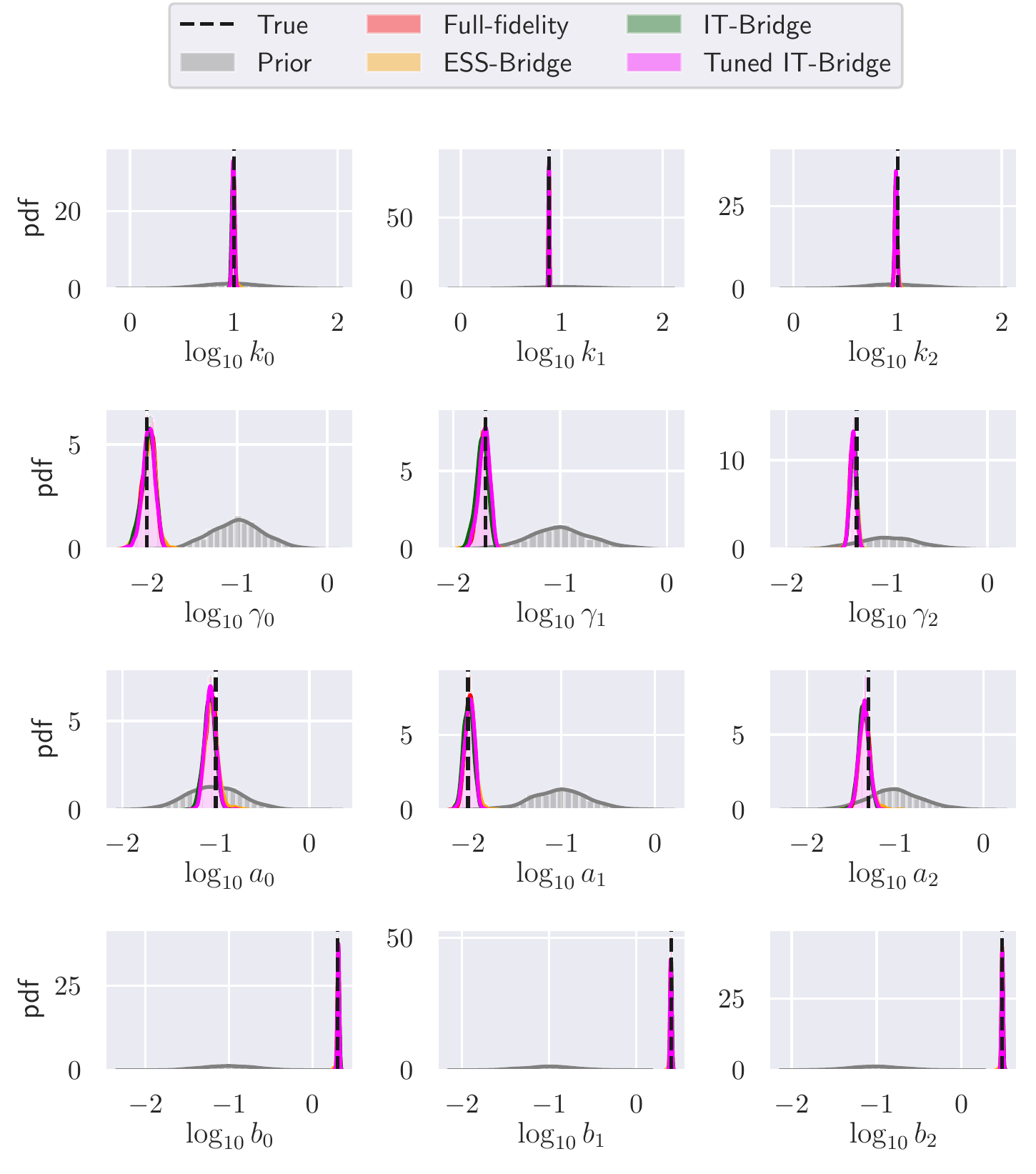}
\caption{Prior and posterior densities in the repressilator example.
 See table~\ref{table:repressilator_posterior} for the numerical values of the estimated means and standard deviations of these posterior distributions. It is evident that all methods converge to virtually the same distribution.}
  \label{fig:repressilator_posteriors}
\end{figure}

\begin{figure}
  \centering
  \includegraphics[scale=0.85]{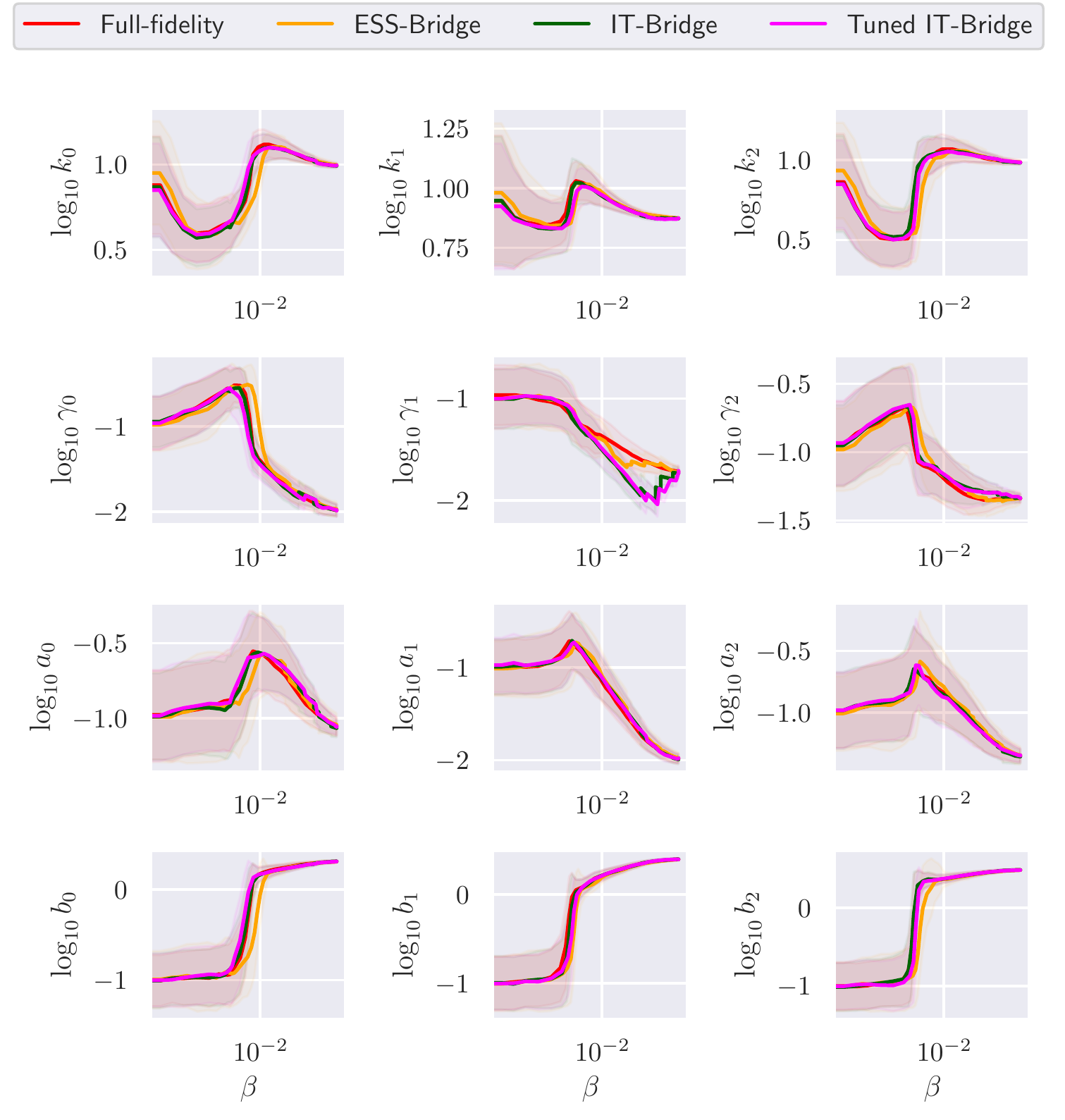}
  \caption{Evolution of the population of samples for the repressilator model parameters using four different ST-MCMC variants: full-fidelity, multifidelity strategies with bridging based on ESS, Information Theoretic Criteria and Tuned Information Theoretic Criteria. The solid lines represent the history of the sample means. The area of the mean $\pm$ standard deviation is presented in the shaded region. Notice in the $\gamma_1$ parameter that bias starts to accumulate for the IT-based methods. This bias is corrected when the sampler starts bridging since the bias began to exceed the natural parameter variability. }
  \label{fig:repressilator_sampling_history}
\end{figure}

\subsection{Bayesian comparison of comparmental models of gene expression}

\begin{table}
  \centering
  \begin{tabular}{llr}
      \toprule
      reaction index & reaction & propensity \\
      \midrule
      $1,\ldots,n_{G} $
      & $G_{i-1} \rightarrow G_{i}$, $i=1,\ldots, n_{G}-1$
      & $k_{i-1}^{+}[G_{i-1}]$
      \\
      $n_{G}+1,\ldots,2n_{G}$
      & $G_{i} \rightarrow G_{i-1}$, $i=1,\ldots,n_{G}-1$
      & $k_{i}^{-}[G_{i}]$
      \\
      $2n_{G}+1, \ldots, 3n_{G}-1$
      & $G_i \rightarrow G_i + \text{RNA}_{nuc}$ , $i = 1,\ldots, n_{G}-1$
      & $r_{i}[G_i]$
      \\
      $3n_{G}$ & $\text{RNA}_{nuc}\rightarrow \text{RNA}_{cyt}$
      & $k_{trans}[\text{RNA}_{nuc}]$
      \\
      $3n_{G}+1$ & $\text{RNA}_{cyt} \rightarrow \emptyset$ & $\gamma[\text{RNA}_{cyt}]$
      \\
      \bottomrule
  \end{tabular}
  \caption{
    Reactions and propensities in the compartmental gene expression model.
  }
  \label{table:gene_expression_reactions}
\end{table}

We next explore the application of multifidelity ST-MCMC to the problem of model selection.
We consider a class of compartmental multi-state gene expression models based on the model considered in~\cite{Munsky2018}.
The model separates biomolecules into the nuclear and cytoplasmic compartments.
The reaction network consists of a gene that could switch between an inactivated state $G_0$ and several activated states $G_i$, $i=1,\ldots,n_{G}-1$. When activated, these gene can be transcribed into RNA molecules within the nucleus at the rate of $r_{i}$ molecule/minute on average.
These nuclear mRNA molecules are then transported into the cytoplasm at a rate of $k_{trans}$ molecule/min, where they degrade at the probabilistic rate $\gamma$ molecule/minute.
Overall, the model consists of $n_{G} + 2$ species: genes that are at different states, nuclear mRNA and cytoplasmic mRNA. These molecular species that can go through $3n_G + 1$ reaction channels~(Table~\ref{table:gene_expression_reactions}).
Only the copy numbers of the nuclear and cytoplasmic mRNA species are observable in experiments. We want to use model selection to decide the number $n_{G}$ of gene states that best explain the observed data.

We simulate a ground truth dataset based on the model with $n_{G}:=3$, which consists of $1000$ single-cell measurements for each time point $t\in \{ 2, 4, 6, 8, 10 \}$ (with hour as time unit). We then use the multifidelity ST-MCMC using the information theoretic criteria with $\beta$-tuning to estimate the model evidence for three classes of reaction networks that consist of two, three, and four gene states and compare these results with the model evidence founding using the full model. We choose the information theoretic criteria with $\beta$-tuning over the other multifidelity approaches because it is the most robust and uses the $\beta$-tuning to avoid sampling degeneracy when the model fidelity changes. This is very important for computing model evidence because the model evidence estimate error is related to the KL-divergence between the intermediate distributions. The 2 and 3 state model were run using 7 nodes with 36 cores, while the 4 state model was run on 14 nodes also with 36 cores each. Each ST-MCMC used 1008 parallel samples. For each level, chains were run using MCMC until a correlation target of $0.6$ was reached or $100$ iterations exceeded. These results are summarized in Table~\ref{table:gene_evidence}. The model evidence ranges are computed using the approach described in~\cite{CALDERHEAD20094028}.

From our results, we observe that not only does the Tuned-IT multifidelity ST-MCMC provide consistent estimate of the model evidence compared to the full-model based ST-MCMC, it actually predicts less error. All the while taking less time with speed up factors of about $1.6$, $3.8$ and $3.2$ for the 2, 3, and 4 gene state model respectively. The improved estimate of the multifidelity approach is likely due to the fact that it uses more intermediate levels so it has a finer discretionary of the thermodynamic integration used to estimate the evidence.

The computed evidence indicates that the present data does not significantly favor one model choice over others. Clearly more experiments are required to provide conclusive evidence for model selection, and the multifidlelity framework allows us to realize the insufficiency of data faster than the full-fidelity scheme. This is an important advantage in practice, as the faster assessment of current experimental data will likely reduce the lag time between consecutive batches of experiments. Further it may be possible to integrate the same multifidelity ST-MCMC based approach into Bayesian experimental design to speed up estimating the expected information gain from various experimental setups in order to design experiments to better discriminate between the models.

\begin{table}
  \centering
    \begin{tabular}{l|rr|rr} 
\toprule\makecell[c]{Model} & \multicolumn{2}{c}{Full-fidelity} & \multicolumn{2}{c}{Tuned IT Bridge} \\
& Log Evidence & Time (Sec)  & Log Evidence & Time (Sec)  \\ 
\hline 2 Gene & $-20108.8 \pm  5.4$ & $1244$  & $ -20112.6 \pm 2.0$  & $ 758$  \\ 
          3 Gene & $-20111.9 \pm  5.6$  & $67496$  	 & $ -20115.7 \pm 2.0$   & $ 17511$  \\ 
          4 Gene & $-20113.0 \pm  5.7$  & $76546$  		 & $ -20117.5 \pm 1.8$   & $ 23777$  \\  
\bottomrule\end{tabular}
  \caption{Comparison of the model evidence computation for the 2,3, and 4 state gene expression model using ST-MCMC with the full fidelity and Multifidelity ST-MCMC with $\beta$-tuning. The evidence estimates from the full-fidelity and multifidelity methods are consistent, but the multifidelity scheme is significantly faster. }
  \label{table:gene_evidence}
\end{table}

\subsection{Stochastic transcription of the inflammation response gene IL1beta}

\begin{table}
  \centering
  \begin{tabular}{lll}
      \toprule
      & reaction & propensity \\
      \midrule
      1.
      & $G_0 \rightarrow G_1$
      & $k_{01}[G_0]$
      \\
      2.
      & $G_{1} \rightarrow G_{2}$
      & $k_{12}[G_1]$
      \\
      3.
      &
      $G_{2} \rightarrow G_{1}$
      & $k_{21}[G_2]$
      \\
      4.
      &
      $G_{1} \rightarrow G_{0}$
      &
      $k_{10}(t) = \max\{0, a_{10} - b_{10}S(t)\}$, see eq.\eqref{eq:il1b_signal}
      \\
      5.
      &
      $\emptyset \rightarrow \text{RNA}$
      &
      $\alpha_1[G_1] + \alpha_2[G_2]$
      \\
      6.
      &
      $\text{RNA} \rightarrow \emptyset$
      &
      $\gamma[\text{RNA}]$
      \\
      \bottomrule
  \end{tabular}
  \caption{
    Reactions and propensities in the IL1beta model.
  }
  \label{table:il1b_reactions}
\end{table}

Having explored the performance of the Multifidelity ST-MCMC schemes with FSP on theoretical examples with simulated datasets, we apply our method on modeling real datasets.
We consider the expression of the IL1beta gene in response to LPS stimulation that was studied in Kalb et al.~\cite{Kalb2019}.
The dataset consists of mRNA counts for IL1beta measured right before applying LPS stimulation, as well as those at $[0.5, 1, 2, 4]$ hours after.
We consider a three-state gene expression model with a time-varying deactivation rate. This results in a chemical reaction network with time-varying propensities and eleven uncertain parameters (Table~\ref{table:il1b_reactions}).
We assume the initial state $(2,0,0,0)$. The observed mRNA counts are fit to the solutions of the CME at times $T_0 + \{0, 0.5, 1, 2, 4\}$ hour, where the time offset $T_0$ is to be estimated.
The influence of LPS-induced signaling molecules is modeled by the function of the form
\begin{equation}
  \label{eq:il1b_signal}
  S(t) = \max\left\{0, \exp\left(-r_1 (t - T_0)\right)\left(1 - \exp\left(-r_2 (t-T0)\right)\right) \right\}.
\end{equation}
This signal affects the rate by which the gene turns off,
$$
k_{1,0}(t) = \max\{0, a_{10} - b_{10}S(t)\}.
$$
Similar to the previous example, the gene state is hidden and the data only contains measurements of the mRNA copy numbers.

Fig.~\ref{fig:il1b_performance} summarizes the performance of the four sampling schemes.
The full-fidelity model considered in this example has a state space of only $18,432$ states, and could be solved quickly without any reduction scheme. Yet, we still observe significant speedup from the Multifidelity schemes.
Specifically, the Multifidelity ST-MCMC with ESS, Information Theoretic and Tuned Information Theoretic Bridging took respectively 3153, 2384, and 2703 seconds to finish, with speedup factors of 1.7, 2.2, and 2.0 over the full-fidelity scheme that takes over 5468 seconds. In Figure~\ref{fig:il1b_evo}, we see the evolution of the model parameters for the different methods and we can use this to better understand differences from Figure~\ref{fig:il1b_performance}. It took significantly longer for the ESS-based method to bridge to  fidelity models than the IT-based methods. However, when it did bridge it jumped straight to the highest model fidelity. For several parameters, we see that their evolution under the ESS-based scheme accumulated significant bias at lower $\beta$ levels before being corrected when the ESS sampler started bridging much later on (e.g. parameters $r_1$, $k_{01}$, and $T_0$). Furthermore, once bridging occurred the distributions were very far apart from each other so the sample population degenerated. This can be seen in the bias that occurs in parameters $r_2$, $b_{10}$ and $\alpha_1$ immediately after bridging. The fact that the ESS-based method does not use any information from the full posterior explains this delay and degeneracy as it is unable to detect the emergence of bias. In contrast, the IT-based methods use the guidance of full model evaluations to better recognize the emergence of bias, so they correct it quicker. Therefore, the IT-based schemes have a smoother evolution and as a result take less time. However, we do observe bias occur in $r_2$ after one bridging step for the IT-Bridge without $\beta$ tuning, indicating some sampling degeneracy. This is not the case for IT-Bridge with $\beta$ tuning since it is specifically designed to avoid degeneracy.

All sampling schemes use essentially the same posterior estimates for the model parameters (Table~\ref{table:il1b_posterior}). Despite significant posterior variance for some parameters, the Bayesian prediction for the distributions of RNA copy number has negligible uncertainties, and they appear to correspond reasonably well with the experimental data at the beginning and the end of the measurement time period (Fig.~\ref{fig:il1b_prediction}).
We notice that this is not necessarily the only model structure that could explain the data, and there may yet be other models that could fit and predict single-cell behavior more accurately.
The speedup enabled by the Multifidelity framework will allow the researcher the ability to more rapidly propose, assess, and choose between different alternative models.

\begin{table}
  \centering
    \begin{tabular}{l|r|rrrr} 
\toprule\makecell[c]{Parameter}& \makecell[c]{Prior} & \multicolumn{4}{c}{Posterior} \\&& Full-fidelity& ESS-Bridge& IT-Bridge& Tuned IT-Bridge\\ 
\hline$\log_{10}(r_1)$ & $-2.00 \pm  0.33$ & $ \text{-2.48} \pm 0.03$  & $ \text{-2.47} \pm 0.03$  & $ \text{-2.47} \pm 0.03$  & $ \text{-2.47} \pm 0.03$ \\ 
$\log_{10}(r_2)$ & $-2.00 \pm  0.33$ & $ \text{-2.00} \pm 0.34$  & $ \text{-2.00} \pm 0.33$  & $ \text{-2.00} \pm 0.33$  & $ \text{-2.00} \pm 0.32$ \\ 
$\log_{10}(k_{{01}})$ & $-3.00 \pm  0.33$ & $ \text{-3.26} \pm 0.03$  & $ \text{-3.25} \pm 0.04$  & $ \text{-3.25} \pm 0.03$  & $ \text{-3.25} \pm 0.03$ \\ 
$\log_{10}(a_{{10}})$ & $-2.00 \pm  0.33$ & $ \text{-1.31} \pm 0.06$  & $ \text{-1.31} \pm 0.06$  & $ \text{-1.30} \pm 0.06$  & $ \text{-1.31} \pm 0.06$ \\ 
$\log_{10}(b_{{10}})$ & $ 3.00 \pm  0.33$ & $ \text{3.04} \pm 0.34$  & $ \text{3.05} \pm 0.34$  & $ \text{3.06} \pm 0.30$  & $ \text{3.04} \pm 0.31$ \\ 
$\log_{10}(k_{{12}})$ & $-3.00 \pm  0.33$ & $ \text{-3.08} \pm 0.03$  & $ \text{-3.08} \pm 0.04$  & $ \text{-3.08} \pm 0.04$  & $ \text{-3.08} \pm 0.03$ \\ 
$\log_{10}(k_{{21}})$ & $-2.00 \pm  0.33$ & $ \text{-1.25} \pm 0.15$  & $ \text{-1.26} \pm 0.14$  & $ \text{-1.28} \pm 0.14$  & $ \text{-1.28} \pm 0.13$ \\ 
$\log_{10}(\alpha_1)$ & $-3.00 \pm  0.33$ & $ \text{-3.60} \pm 0.15$  & $ \text{-3.60} \pm 0.16$  & $ \text{-3.60} \pm 0.16$  & $ \text{-3.60} \pm 0.16$ \\ 
$\log_{10}(\alpha_2)$ & $ 0.00 \pm  0.33$ & $ \text{0.71} \pm 0.15$  & $ \text{0.70} \pm 0.14$  & $ \text{0.68} \pm 0.13$  & $ \text{0.68} \pm 0.13$ \\ 
$\log_{10}(\gamma)$ & $-4.00 \pm  0.33$ & $ \text{-4.56} \pm 0.05$  & $ \text{-4.56} \pm 0.05$  & $ \text{-4.56} \pm 0.05$  & $ \text{-4.56} \pm 0.05$ \\ 
$\log_{10}(T_0)$ & $ 4.00 \pm  0.33$ & $ \text{5.36} \pm 0.09$  & $ \text{5.36} \pm 0.09$  & $ \text{5.37} \pm 0.08$  & $ \text{5.37} \pm 0.09$ \\ 
\bottomrule\end{tabular}
  \caption{Model parameters in the IL1beta example. The second column presents the parameters of the prior distribution, where we use a Gaussian prior in the $\log_{10}$-transformed parameter space with a diagonal covariance matrix. The last four columns present the posterior mean and standard deviation of model parameters estimated using the ST-MCMC with full-fidelity model and the Multifidelity ST-MCMC with three different bridging strategy.}
  \label{table:il1b_posterior}
\end{table}

\begin{figure} [H]
  \centering
  \includegraphics[scale=0.75]{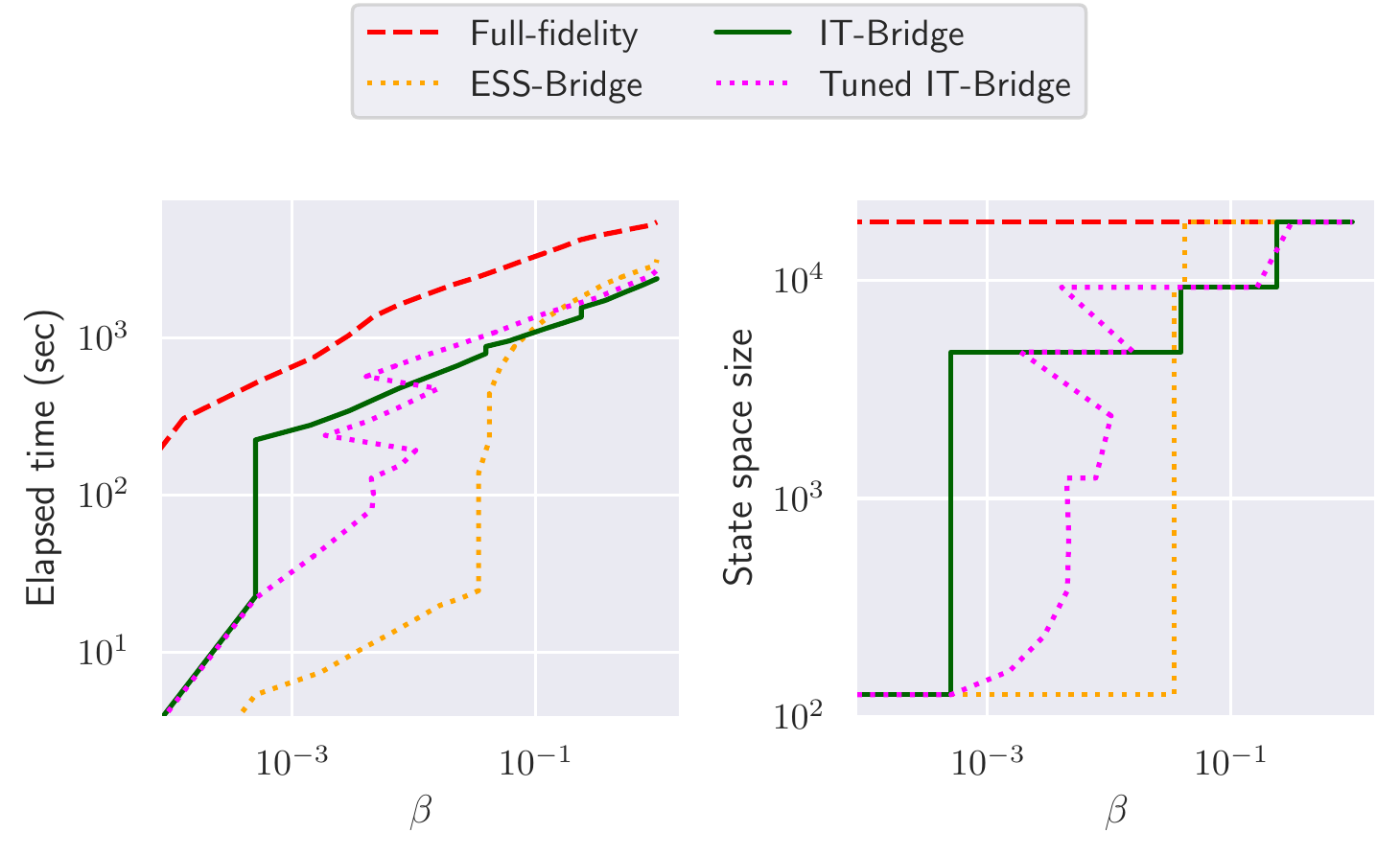}
  \caption{Performance of STMCMC samplers on the IL1beta example. The horizontal axis represents the inverse temperature.}
  \label{fig:il1b_performance}
\end{figure}

\begin{figure}
  \centering
  \includegraphics[scale=0.85]{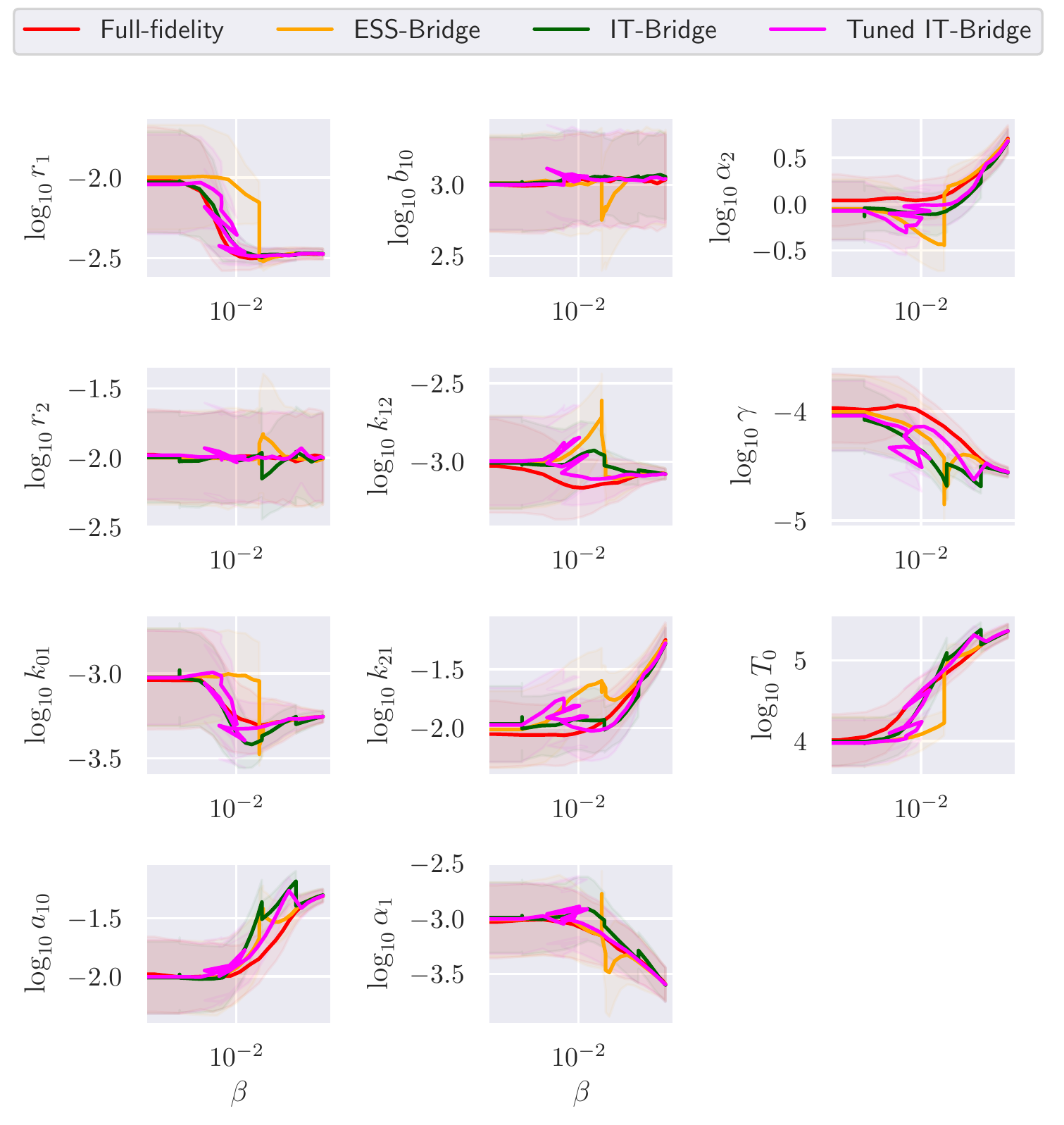}
  \caption{Evolution of the population of samples for the IL1beta model parameters using four different ST-MCMC variants: full-fidelity, multifidelity strategies with bridging based on ESS, Information Theoretic Criteria and Tuned Information Theoretic Criteria. The solid lines represent the history of the sample means. The area of the mean $\pm$ standard deviation is presented in the shaded region.}
  \label{fig:il1b_evo}
\end{figure}

\begin{figure}
  \centering
  \includegraphics[scale=0.75]{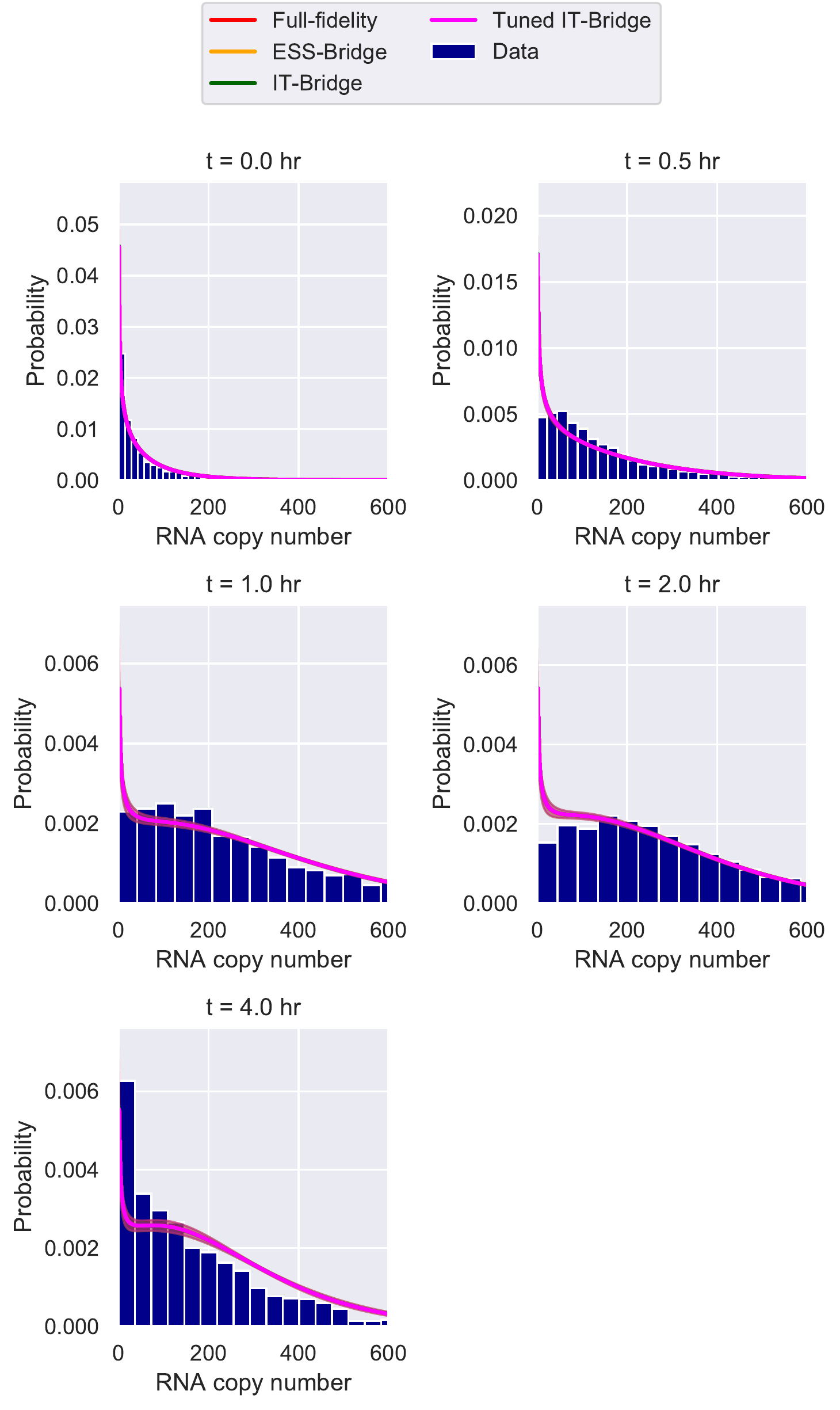}
  \caption{Comparison of data and the posterior mRNA distribution predictions for the IL1beta transcription model at zero and four hour after LPS induction. The mean Bayesian prediction for the mRNA probability distribution is computed by averaging the solution of the CME over all posterior samples. The area of one standard deviation around the mean is shown in shade. Visually speaking, samples from different ST-MCMC formulations yield identical predictions.}
  \label{fig:il1b_prediction}
\end{figure}

\section{Conclusion}
Rapid advancements in experimental techniques are allowing biologists to collect quantitative data about cellular processes at ever smaller scales with increasing detail~\cite{Raj2008, Munsky2015, Li2019}.
Mathematical models have become an indispensable part in the process of learning and making predictions from this data.
Stochastic reaction networks (SRNs) form a powerful class of models that have found widespread use within the quantitative biology community~\cite{Munsky2015}.
Identifying these models from the data, however, is a challenging task due to the computational cost of solving the chemical master equation (CME).
This has prevented a fully Bayesian statistical framework from being adopted widely in real biological studies.
In this paper, we seek to address the challenge of applying the Bayesian philosophy to analyzing stochastic gene expression data by proposing an efficient computational framework for Bayesian parameter calibration and model selection for SRNs. This framework combines novel multifidelity formulations of the massively parallel ST-MCMC sampler with surrogate models of the CME.
Numerical tests demonstrate that this combined approach leads to significant savings in comparison to a state-of-the-art method that uses solely the high-fidelity models. Further, we also propose a new criteria for tuning model fidelity within multifidelity SMC type methods based on information theory that compares favorably to effective sample size based techniques.

The research reported here may potentially lead to fruitful future directions. With respect to surrogate models, the approach proposed here for the efficient solution of the surrogate master equations is only one among various alternatives that have been proposed over the years since the introduction of the FSP algorithm~\cite{Munsky2006}. Another attractive option for constructing multifidelity models is to utilize a low-rank tensor format such as the quantized tensor train that has been proposed for the forward solution of the CME~\cite{Kazeev2013, Kazeev2015, Dolgov2014, Dolgov2015, Vo2017c, Dolgov2018}. It is also possible to exploit bounds on the log-likelihood function as done in Fox et al.~\cite{Fox2016}.

The improved efficiency may lead to more widespread adoptions of the Bayesian approach in answering biological questions.
We refer to Catanach et al.~\cite{Catanach2018Context} for an example of a Bayesian approach to studying the phenomenon of context dependence in synthetic gene circuits using Bayesian model selection, which required significant computational resources.

There are also many directions for improving Multifidelity ST-MCMC in general. We expect estimating the information gain criteria could be significantly improved. One possibility is using a more advanced sampling scheme that leverages model evaluations from across the multifidelity hierarchy. This could even further reduce the number of full model evaluations needed at each level. Another avenue of research is designing Multifideltiy ST-MCMC specifically for estimating a quantify of interest to a given accuracy as is done in Multilevel MCMC. If we have a design object, ST-MCMC may not need to progress through the full model hierarchy or all annealing levels in order to provide enough information to estimate the quantify of interest to the desired accuracy. By further reducing the computational cost of Bayesian methods like ST-MCMC, engineers and scientist will be better able to integrate uncertainty quantification into their workflow. Therefore, as high performance computing resources are becoming increasingly accessible, we expect the Multifidelity ST-MCMC framework to provide a useful tools for researchers who are interested in model calibration and uncertainty propagation for complex models.

\label{sec:conclusion}
\section*{Acknowledgement}
We thank James Werner and Daniel Kalb for kindly sharing with us the data from their smFISH experiment. The cited work~\cite{Kalb2019} was performed, in part, at the Center for Integrated Nanotechnologies, an Office of Science User Facility operated for the U.S. Department of Energy (DOE) Office of Science by Los Alamos National Laboratory (Contract 89233218CNA000001) and Sandia National Laboratories (Contract DE-NA-0003525). The work presented here was also funded in part by the Department of Energy Office of Advanced Scientific Computing Research through the John von Neumann Fellowship. Sandia National Laboratories is a multimission laboratory managed and operated by National Technology and Engineering Solutions of Sandia, LLC., a wholly owned subsidiary of Honeywell International, Inc., for the U.S. Department of Energy’s National Nuclear Security Administration under contract DE-NA-0003525. This paper describes objective technical results and analysis. Any subjective views or opinions that might be expressed in the paper do not necessarily represent the views of the U.S. Department of Energy or the United States Government. SAND2019-15382 J.

\noindent We also thank Ania-Ariadna Baetica for providing constructive comments on the manuscript and Jed Duersch for discussions regarding information theory.

\bibliographystyle{preprint}
\bibliography{library}
\end{document}